\pdfoutput=1
\documentclass[10pt,  conference, letterpaper]{IEEEtran}

\IEEEoverridecommandlockouts

\usepackage[table]{xcolor}

\usepackage{amsmath, amsthm, amssymb, amsfonts}
\usepackage{dsfont}
\usepackage{graphicx}
\usepackage{subcaption}
\usepackage{tikz}
\usepackage{pgfplots}
\pgfplotsset{compat=newest}
\usetikzlibrary{
  arrows, arrows.meta, shapes, positioning, calc,
  shadows, shadows.blur, automata,
  decorations.pathmorphing, backgrounds, fit
}

\usepackage{array}
\usepackage{tabularx}
\usepackage{booktabs}
\usepackage{multirow}
\usepackage{threeparttable}
\usepackage{footnote}
\makesavenoteenv{tabularx}

\usepackage{enumitem}
\usepackage{geometry}
\usepackage{fancyhdr}
\usepackage{multicol}
\usepackage{lipsum}

\usepackage{cite}
\usepackage{hyperref}

\usepackage{algorithm}
\usepackage[noend]{algpseudocode}

\usepackage{authblk}

\newcommand{\tr}{\operatorname{tr}}

\newtheorem{theorem}{Theorem}
\newtheorem{lemma}{Lemma}

\newtheorem{definition}{Definition}
\newtheorem{assumption}{Assumption}
\newtheorem{remark}{Remark}

\tikzstyle{server} = [rectangle, rounded corners, minimum width=6cm, minimum height=1.5cm, text centered, draw=green, fill=green!30, text width=5.5cm]
\tikzstyle{worker} = [rectangle, rounded corners, minimum width=2cm, minimum height=0.7cm, text centered, draw=blue, fill=blue!30]
\tikzstyle{ellipsis} = [draw, circle, minimum size=0.5cm]
\tikzstyle{arrow} = [thick,->,>=stealth]
\tikzstyle{highlight} = [ellipse, dashed, draw=black, minimum width=5cm, minimum height=1cm, yshift=-1cm]

\definecolor{senderblue}{RGB}{65, 105, 170}
\definecolor{senderlight}{RGB}{220, 235, 250}
\definecolor{receivergreen}{RGB}{46, 139, 87}
\definecolor{receiverlight}{RGB}{220, 245, 230}
\definecolor{channelgold}{RGB}{218, 165, 32}
\definecolor{channellight}{RGB}{255, 248, 220}
\definecolor{couponorange}{RGB}{230, 126, 34}
\definecolor{couponlight}{RGB}{253, 235, 208}
\definecolor{dangered}{RGB}{192, 57, 43}
\definecolor{successgreen}{RGB}{39, 174, 96}
\definecolor{bglight}{RGB}{252, 252, 254}

\renewcommand{\footnoterule}{%
  \kern -3pt
  \hrule width \columnwidth height 0.4pt
  \kern 2.6pt
}
\author[1]{Youssef Ahmed}
\author[2]{Arnob Ghosh}
\author[3]{Chih-Chun Wang}
\author[1,4]{Ness B. Shroff}

\affil[1]{Department of ECE, The Ohio State University, Columbus, OH, USA}
\affil[2]{Department of Electrical and Computer Engineering,
New Jersey Institute of Technology, Newark, NJ, USA}
\affil[3]{Elmore Family School of ECE, Purdue University, West Lafayette, IN, USA}
\affil[4]{Department of CSE, The Ohio State University, Columbus, OH, USA}

\affil[ ]{\textit{Emails:} ahmed.943@osu.edu, arnob.ghosh@njit.edu, chihw@purdue.edu, shroff.11@osu.edu}

\begin{document}
\title{Beyond Freshness and Semantics: A Coupon-Collector Framework for Effective Status Updates \thanks{
This work was supported in part by the National Science Foundation under 
Grants CNS-2107363, CCF-2309887, and ECCS-2418106, and by a gift from 
MediaTek USA. This work has also been supported in part by the Army 
Research Laboratory under Cooperative Agreement Number W911NF-23-2-0225, 
by the U.S. National Science Foundation under the grants: NSF AI Institute 
(AI-EDGE) 2112471, CNS-2312836, CNS-2225561, and CNS-2239677, and by the 
Office of Naval Research under grant N00014-24-1-2729. The views and 
conclusions contained in this document are those of the authors and should 
not be interpreted as representing the official policies, either expressed 
or implied, of the Army Research Laboratory or the U.S. Government. The 
U.S. Government is authorized to reproduce and distribute reprints for 
Government purposes notwithstanding any copyright notation herein.
}}
\maketitle
\begin{abstract}
For status update systems operating over unreliable energy-constrained wireless channels, we address Weaver's long-standing \emph{Level-C} question~\cite{weaver1953recent}: \emph{do my packets actually improve the plant's behavior?}. Each fresh sample carries a stochastic expiration time---governed by the plant's instability dynamics---after which the information becomes useless for control. Casting the problem as a coupon-collector variant with expiring coupons, we (i)~formulate a two-dimensional average-reward MDP, (ii)~prove that the optimal schedule is \emph{doubly thresholded} in the receiver's freshness timer and the sender's stored lifetime, (iii)~derive a closed-form policy for deterministic lifetimes, and (iv)~design a Structure-Aware Q-learning algorithm (SAQ) that learns the optimal policy without knowing the channel success probability or lifetime distribution. Simulations validate our theoretical predictions: SAQ matches optimal Value Iteration performance while converging significantly faster than baseline Q-learning, and expiration-aware scheduling achieves up to 50\% higher reward than age-based baselines by adapting transmissions to state-dependent urgency---thereby delivering Level-C effectiveness under tight resource constraints.
\end{abstract}
\section{Introduction}
Modern cyber-physical systems, such as mobile robots navigating busy spaces or learning agents steering autonomous vehicles, rely on timely state information that must traverse unreliable, energy-constrained wireless links~\cite{shi2016edge, zhou2019edge}. Sending every update guarantees the controller acts on the freshest data, but quickly drains battery life and clogs shared channels~\cite{raghunathan2002energy}. Skipping updates, however, conserves energy while forcing the controller to rely on stale information, increasing the risk of suboptimal or unsafe decisions. Existing frameworks frequently waste energy transmitting ``fresh but useless'' updates because they cannot identify when local knowledge suffices. They lack a measure of ``Level-C'' effectiveness \cite{weaver1953recent}, i.e., the ability to determine when a specific sample truly improves decision quality.

In contrast, computation costs have plummeted as edge devices benefit from continual advances in transistor scaling and specialized machine-learning accelerators \cite{sze2017efficient}. Modern hardware has shifted the energy bottleneck decisively: while performing a full neural-network inference on TinyML hardware now consumes as little as $1.9\,\mu\text{J}$ \cite{wu2023bulk}, transmitting a single data packet over a standard radio protocol can consume orders of magnitude more energy. For instance, transmitting just one bit requires roughly $0.015\,\mu\text{J}$. This means that sending a standard packet (approx.\ 1000 bits including overhead) consumes roughly $15\,\mu\text{J}$ \cite{raghunathan2002energy}. Hence, a device can execute an entire complex inference algorithm locally for significantly less energy than it takes to wirelessly report a single raw measurement. This disparity establishes communication, (not computation) as the primary bottleneck on system performance and lifetime. Using this computation power, one might do effective inference without relying on fresh information. Therefore, we need an intelligent policy that leverages this abundant computational capacity to filter data locally, transmitting updates only when they measurably improve control performance.

Given this asymmetry, the natural question is: how should a sensor decide when to transmit? Classically, this need for local filtering has been addressed by \emph{Event-Triggered Control (ETC)} \cite{tabuada2007event}, which trigger transmissions only when stability is threatened or error thresholds are crossed. However, these methods typically require rigid analytical models of the plant and channel. To enable more flexible, data-driven scheduling, the community has moved beyond the classic network–level objectives of end-to-end delay, throughput, and jitter. Instead, we now optimize \emph{freshness-centric performance metrics}, we refer to as \textbf{Information Timeliness Metrics (ITM)}. The flagship member is the \emph{Age of Information (AoI)} \cite{kaul2012real}, which measures, at any instant, the time elapsed since the newest successfully received sample was generated. While AoI captures data staleness directly, it is content-agnostic; simply minimizing age does not always guarantee control accuracy or efficient resource usage \cite{sun2017update}. This limitation has spurred extensions that refine timeliness by folding in task relevance. The \emph{Value-of-Information (VoI)}, for instance, lacks a single canonical definition; it is modeled either as a non-linear function of age (e.g., via exponential decay) \cite{8764465} or based on the intrinsic system dynamics to reflect how an update impacts state estimation \cite{ayan2019age,arafa2024age}. Similarly, the \emph{Age of Incorrect Information (AoII)} accumulates staleness only when the receiver's estimate deviates effectively from the true state \cite{maatouk2020age}, while the \emph{Age of Synchronization (AoSync)} tracks the time since the controller and plant last shared a fully synchronized state \cite{zhong2018two}.

However, timeliness metrics are often merely proxies for true operational goals. Consequently, researchers frequently optimize Control-Oriented Metrics directly, such as Mean Squared Error (MSE) in networked estimation \cite{sun2017remote}. Bridging timeliness and error, the \emph{Urgency of Information (UoI)} captures the weighted cost of estimation inaccuracy based on context-dependent factors \cite{zheng2020urgency}. Freshness concepts have also been extended to the control link via the \emph{Age of Actuation (AoA)} \cite{nikkhah2023age}, while the \emph{Age of Loop (AoL)} \cite{de2021age} unifies both directions by capturing end-to-end freshness across the closed loop---with downlink-initiated (DL-AoL) and uplink-initiated (UL-AoL) variants. Table~\ref{tab:metrics_map} summarizes these metrics alongside their corresponding Shannon--Weaver communication layers. Collectively, these approaches underscore that effective policies must balance raw freshness with actual estimation and decision-making quality.

\begin{table}[t]
\centering
\scriptsize
\caption{Key update metrics mapped to Shannon-Weaver layers. 
\textit{Symbols:} $n$: current slot; $U(n)$: generation time of latest received update; $\Delta(n)=n-U(n)$; $x_n, \hat{x}_n$: true/estimated states; $A(n)$: last action time; $\hat{t}_j$: freshest command time; $t_i$: status update time causing latest action; $S(n)$: last sync time; $U_{\mathrm{comp}}(n)$: generation time of latest computation result; $w(n)$: urgency weight; $C(\cdot)$: cost function; $r_k$: reward at age $k$; $\epsilon$: tolerance.}
\label{tab:metrics_map}
\renewcommand{\arraystretch}{1.25}
\begin{threeparttable}
\begin{tabularx}{\columnwidth}{
  >{\raggedright\arraybackslash}p{1.2cm}
  >{\raggedright\arraybackslash}X
  >{\centering\arraybackslash}p{1.0cm}
  >{\raggedright\arraybackslash}p{2.0cm}}
\toprule
\textbf{Metric} & \textbf{Definition} & \textbf{Level} & \textbf{Insight} \\
\midrule
AoI & $\Delta(n)=n-U(n)$ & A & Data freshness \\[2pt]
MSE & $\mathbb{E}[\|x_n-\hat{x}_n\|^2]$ & A & Estimation accuracy \\[2pt]
AoSync & $n-S(n)$ & A & Synchronization gap \\[2pt]
AoA & $n-A(n)$ & A & Actuation staleness \\[2pt]
AoL & $t-\hat{t}_j$; $t-t_i$ & A & Round-trip freshness \\[2pt]
AoC & $n - U_{\mathrm{comp}}(n)$ & A & Computation freshness \\[2pt]
VoI & $f(\Delta(n),x_n)$ & A/B\tnote{1} & Context-aware utility \\[2pt]
UoI & $w(n)\cdot C(x_n,\hat{x}_n)$ & B & Weighted inaccuracy cost \\[2pt]
AoII & $\sum_{k=U(n)}^{n}\mathds{1}\{x_k\neq\hat{x}_k\}$ & B & Error-aware staleness \\[2pt]
\textbf{Expiration (Ours)} & $\max\{k:r_0-r_k\leq\epsilon\}$ & C & Just-in-time scheduling \\
\bottomrule
\end{tabularx}
\begin{tablenotes}
\scriptsize
\item[1] Depends on VoI definition used.
\end{tablenotes}
\end{threeparttable}
\vspace{-2mm}
\end{table}
This shift from signal fidelity to functional performance places our work within the emerging paradigm of \textbf{Goal-Oriented Communication} (Level-C) \cite{strinati20216g, uysal2022semantic, li2024toward}. While classical information theory addresses \emph{Level A} (accurate symbol transmission) and Semantic Communication targets \emph{Level B} (meaning preservation), Goal-Oriented strategies focus on \emph{Level C}: ensuring that the received information successfully steers the system toward a specific objective. The defining feature of this level is that ``all information not strictly relevant to the fulfillment of the goal can be neglected.'' Our framework implements precisely this principle: by filtering updates based on their expiration time (a proxy for their relevance to the control objective) we discard technically correct but functionally irrelevant data, prioritizing the \emph{pragmatic value} of information over its mere existence.

Building on this principle, we develop a status-update framework that explicitly models how each transmission (or silence) influences the decision making process. In particular, we consider a sender-receiver pair in which a fresh sample arrives every time slot, but each sample is only \emph{useful} to the controller for a finite, random \emph{expiration time}. The sender must decide, under an unreliable channel and a non-negligible transmission cost, whether the prospective benefit of informing the controller outweighs the cost of sending the packet. The unreliable channel compounds this decision: stochastic packet losses create uncertainty about successful delivery, forcing the sender to weigh retransmitting a potentially stale packet against transmitting a fresher sample, while each attempt consumes resources regardless of outcome.

We develop the expiration framework in two phases. First, we characterize \emph{when observations expire}: \textbf{Section~\ref{ssec:vector-LQ}} derives expiration times analytically for linear systems, while \textbf{Section~\ref{ssec:cartpole-case}} estimates them via Monte Carlo rollouts for nonlinear systems (CartPole), demonstrating generalization beyond tractable dynamics. Both assume reliable channels to isolate the expiration phenomenon. Second, we address \emph{how to schedule under unreliable channels}: \textbf{Sections~\ref{sec:MDP}--\ref{sec:structural}} formulate the problem as a coupon-collector MDP, prove optimal policies are threshold-based, and develop structure-aware Q-learning for the model-free case. Our main contributions are:

\subsection*{Main Contributions}

\begin{enumerate}[label=\textbf{C\arabic*.}, leftmargin=*]
    \item  We close Weaver's \emph{Level-C} gap by attaching a finite, sample-specific \emph{expiration time} to every status update and asking when a transmission truly improves the controller's action quality. Two representative case studies (a linear Kalman-controlled plant with an MSE cap and a data-driven Remotely controlled Cartpole task) show how such expiration horizons arise analytically or via learning.
    \item 
          We recast scheduling as a coupon-collector problem with expiring coupons and formulate the decision process as a two-dimensional average-cost MDP. Using lattice monotonicity we prove the bias value is non-decreasing in both state coordinates; the action gap is super-modular, yielding a \emph{double-threshold} optimal policy. For deterministic lifetimes we derive a closed-form threshold rule.
    \item  For unknown, random lifetimes we design a model-free Q-learning algorithm that hard-codes the ``never-send-obsolete'' rule, confines exploration to the undecided band, and converges far faster than vanilla Q-learning.
\item Simulations demonstrate that expiration-aware scheduling achieves up to \(\mathbf{50\%}\) higher cumulative reward compared to periodic baselines at moderate communication costs, while SAQ matches optimal Value Iteration performance and converges significantly faster than baseline Q-learning.\end{enumerate}
\section{Characterizing Observation Expiration}
\label{sec:expiration}
Before addressing the scheduling problem under unreliable channels 
(Sections~\ref{sec:MDP}--\ref{sec:structural}), we must first answer 
a fundamental question: \emph{how long does an observation remain useful?} 
This section develops two complementary approaches to characterizing 
expiration time $T$: analytical derivation for linear systems 
(Section~\ref{ssec:vector-LQ}) and data-driven estimation for nonlinear 
systems (Section~\ref{ssec:cartpole-case}). Both assume a reliable channel 
($p_s = 1$), isolating the expiration phenomenon from transmission uncertainty.

%---------------------------------------------------------------
%  Case 1 : Vector Linear–Quadratic Control Loop
%---------------------------------------------------------------
%---------------------------------------------------------------
%===============================================================

\subsection{Linear control system with MSE-based expiration}
\label{ssec:vector-LQ}

\paragraph{Plant and sensor model.}
We study the discrete-time linear system
\begin{equation}\label{eq:plant}
\begin{aligned}
\mathbf x_{n+1} &= A\,\mathbf x_n + B\,u_n + \mathbf w_n,\\
\mathbf y_{n+1} &= H\,\mathbf x_{n+1} + \mathbf v_{n+1},
\end{aligned}
\end{equation}
where the state $\mathbf x_n\!\in\!\mathbb R^{d}$, control
$u_n\!\in\!\mathcal U\subseteq\mathbb R^{m}$, and measurement
$\mathbf y_{n+1}\!\in\!\mathbb R^{p}$.  
System matrices satisfy
$A\!\in\!\mathbb R^{d\times d}$,
$B\!\in\!\mathbb R^{d\times m}$,
$H\!\in\!\mathbb R^{p\times d}$.
The process noise is $\mathbf w_n\!\sim\!\mathcal N(0,Q)$ with $Q\succ 0$.
The measurement noise covariances $\{R_n\}_{n\ge 1}$ are drawn i.i.d.\ 
from a distribution $\mathcal R$ supported on positive-definite matrices, 
with $R_n$ revealed at the start of slot~$n$.

The sensor runs a standard Kalman filter, producing
$(\hat{\mathbf x}_{n+1|n+1},P_{n+1|n+1})$
at the end of slot~$n{+}1$.
If the estimate is \emph{transmitted} ($a_{n+1}=1$) it is copied
verbatim to the controller; otherwise ($a_{n+1}=0$) the controller
propagates its last estimate through the open-loop prediction
$\hat{\mathbf x}^{\mathrm c}_{n+1}=A\hat{\mathbf x}^{\mathrm c}_{n}+B\,u_n$,
$P^{\mathrm c}_{n+1}=AP^{\mathrm c}_{n}A^{\top}+Q$.
We write $P^{\mathrm c}_{n} \triangleq P_{n|m}$ for the controller's 
error covariance at slot~$n$ when the most recent transmission occurred 
at slot~$m$; between updates this covariance grows monotonically 
because of process noise and the unstable dynamics.
The control law is a fixed map
$u_n=\mu\!\bigl(\hat{\mathbf x}^{\mathrm c}_{n},P^{\mathrm c}_{n}\bigr)$.

\paragraph{Time-varying safety criterion.}
The controller cannot access the \emph{true} state; it only sees an
\emph{estimate} corrupted by error covariance $P^{\mathrm c}_{n}$.
We consider a time-varying tolerance sequence $\{\tau_n\}_{n \geq 0}$ 
with $\tau_n > 0$ and require
\begin{equation}\label{eq:MSEcap}
    \tr\!\bigl(P^{\mathrm c}_{n}\bigr)\;\le\;\tau_n, \qquad \forall\, n,
\end{equation}
which upper-bounds the instantaneous mean-squared error the controller 
is willing to bear.

\begin{assumption}[Stationary thresholds]
\label{as:predictable-tau}
The threshold process $\{\tau_n\}_{n \geq 0}$ is stationary and ergodic, 
independent of the measurement noise $\{R_n\}$, and fully known to the 
sensor (i.e., the entire realization $\{\tau_n\}_{n \geq 0}$ is 
available at time~$0$).\footnote{This holds, e.g., for pre-announced 
mission profiles or periodic task structures whose schedule is fixed 
before deployment.}
\end{assumption}

%\begin{assumption}[Stochastic measurement noise]
%\label{as:iid-noise}
%The measurement noise covariances $\{R_n\}_{n \geq 1}$ are drawn i.i.d.\ 
%from a distribution $\mathcal{R}$ supported on positive definite matrices, 
%with $R_n$ revealed at the start of slot $n$.
%\end{assumption}

\begin{assumption}[Feasibility]
\label{as:feasibility}
For all $n$, $\tr(P_{n|n}) \leq \tau_n$ almost surely. That is, 
transmitting always satisfies the instantaneous constraint.
\end{assumption}

At the start of slot $n$, the noise covariance $R_n$ is revealed to 
the sensor. The sensor computes $(P_{n|n-1}, P_{n|n})$ and selects 
$a_n\in \{0,1\}$. The sensor seeks a causal policy $\pi = \{a_n\}$ 
maximising the long-run expected reward:
\begin{equation}\label{eq:objective}
\max_{\{a_n\}}
\liminf_{N\to\infty}\frac{1}{N}
\mathbb{E}\left[\sum_{n=0}^{N-1}
\Bigl[\mathds{1}\{\tr(P^{\mathrm c}_{n})\le\tau_n\}-\beta c\,a_n\Bigr]
\,\Big|\, P_{0|0}\right].
\end{equation}
The first term awards a unit reward at every slot in which the 
constraint $\tr(P^{\mathrm{c}}_{n}) \le \tau_n$ is met, so the 
objective measures the long-run fraction of time the controller 
operates within its error budget minus a transmission penalty.
The weight $\beta > 0$ trades energy against performance: large 
$\beta$ discourages communication, while $\beta \to 0$ recovers the 
``always transmit'' policy.

Our formulation departs from the standard remote estimation paradigm in 
three fundamental respects. \emph{First}, we model measurement noise as time-varying ($R_n$ rather 
than constant $R$), capturing realistic scenarios where sensor fidelity 
depends on environmental conditions, operating modes, or adaptive power 
management. This generalization is absent from the canonical works 
\cite{sinopoli2004kalman, schenato2007foundations, shi2011sensor, 
leong2017transmission, chakravorty2020remote}.
\emph{Second}, we allow the MSE tolerance $\tau_n$ to vary 
stochastically according to a stationary process, rather than enforcing 
a static constraint throughout operation.
\emph{Third}, we optimize the 
\emph{constraint-satisfaction indicator}
$\mathds{1}\{\mathrm{tr}(P_n^{\mathrm{c}}) \leq \tau_n\}$ 
rather than the MSE itself. 
While prior work minimizes aggregate estimation cost 
$\mathbb{E}[\sum_n \mathrm{tr}(P_n^{\mathrm{c}})]$ 
\cite{xu2004optimal, lipsa2011remote} 
or penalizes large errors through convex surrogates 
\cite{leong2020deep}, 
our objective maximizes the \emph{fraction of time} the controller operates within 
its acceptable-accuracy regime, subject to communication costs.
This binary formulation directly reflects safety-critical requirements 
(e.g., ``maintain localization accuracy 95\% of the time'') and 
naturally gives rise to the \emph{expiration time} abstraction: 
each sample remains \emph{valid} precisely until the constraint 
would be violated, after which it contributes zero reward regardless 
of its residual estimation quality.
\begin{definition}[MSE-based expiration time]
\label{def:expiration}
Given a transmission at slot $m$ with post-update covariance $P_{m|m}$, 
the \emph{expiration time} is
\begin{equation}\label{eq:expiration-time}
    T^{\mathrm{MSE}}_m \triangleq \min\bigl\{k \geq 1 : 
    \tr(P_{m+k|m}) > \tau_{m+k}\bigr\},
\end{equation}
with the convention $\min \varnothing = +\infty$.
\end{definition}

At slot $n$, let $m = \max\{k < n : a_k = 1\}$ denote the most recent 
transmission epoch. The \emph{residual validity} 
$\Delta_n \triangleq T^{\mathrm{MSE}}_m - (n - m)$ measures remaining 
slots before the current estimate expires. A causal policy $\pi$ is 
\emph{admissible} if it maintains $\Delta_n \geq 1$ for all $n$, 
ensuring the constraint \eqref{eq:MSEcap} holds almost surely.

The \emph{just-in-time (JIT) policy} $\pi^{\mathrm{JIT}}$ transmits 
exactly when residual validity reaches one: 
$a_n^{\mathrm{JIT}} = \mathds{1}\{\Delta_n = 1\}$.
The following theorem establishes that this elementary policy is 
optimal among all causal policies.

\begin{theorem}[Optimality of JIT policy]
\label{thm:optimal-policy}
Under Assumptions~\ref{as:predictable-tau}--\ref{as:feasibility}, 
if the transmission cost satisfies 
$\beta c \leq 1$ and $\beta c < \bar{T}$, 
where $\bar{T} \triangleq \lim_{K \to \infty} \frac{1}{K} 
\sum_{k=1}^{K} T^{\mathrm{MSE}}_{m_k}$ is the limiting average 
expiration time along the JIT transmission epochs $\{m_k\}$,
then the just-in-time policy is optimal among 
all causal policies (both admissible and non-admissible), achieving
\begin{equation}\label{eq:jit-value}
    J(\pi^{\mathrm{JIT}}) = 1 - \frac{\beta c}{\bar{T}}.
\end{equation}
\end{theorem}

\begin{proof}
See Appendix~\ref{app:proof-optimal-policy}.
\end{proof}

\begin{remark}[Interpretation of the conditions]
\label{rem:conditions}
Because $R_n$ is i.i.d.\ and $\{\tau_n\}$ is stationary 
(Assumption~\ref{as:predictable-tau}), the Kalman filter reaches a 
stationary regime in which $T^{\mathrm{MSE}}_m$ has the same 
marginal distribution at every epoch~$m$.  Hence $\bar{T} = 
\mathbb{E}[T^{\mathrm{MSE}}_m]$, and the condition $\beta c < \bar{T}$ 
simply requires that the cost of one transmission be less than the 
mean number of reward-earning slots it provides.  The condition 
$\beta c \leq 1$ ensures that a policy can never benefit from 
deliberately letting the constraint be violated: every slot that 
violates \eqref{eq:MSEcap} forfeits reward~$1$, while the 
saved transmission cost is at most $\beta c \leq 1$, so the 
trade-off is never favorable.
\end{remark}

\begin{remark}
The JIT policy extracts maximum value from each transmission by 
consuming its entire validity window. Any policy that transmits 
earlier ``wastes'' validity, while any policy that transmits later 
violates the constraint.
Section~\ref{ssec:cartpole-case} estimates expiration times for a 
nonlinear system where no closed form exists.
\end{remark}
\subsection{Case Study 2: Data-Driven Expiration (Deep RL CartPole)}
\label{ssec:cartpole-case}

Many systems of practical interest are nonlinear and lack closed-form 
covariance recursions.  To show that expiration times can still be 
obtained, we study the CartPole balancing 
problem~\cite{barto1983neuronlike}, whose unstable dynamics share 
structural features with rocket landing~\cite{blackmore2010minimum} 
and self-balancing platforms~\cite{grasser2002joe}.  It is frequently 
used to benchmark information-timeliness 
metrics~\cite{zheng2020urgency,de2021age}.

A recurrent PPO controller tracks a moving reference while receiving 
noisy, potentially stale observations from a remote sensor (full 
setup in Appendix~\ref{app:cartpole-details}).  Transmissions are 
reliable ($p_s = 1$) and cost~$c$ each.

Fig.~\ref{fig:combined_performance} quantifies the joint effect of 
observation age and measurement noise on controller performance.  
In the noise-free regime ($\sigma_{\max}=0$), the controller 
maintains 100\% survival for observation ages up to 5~steps 
(Fig.~\ref{fig:survival_heatmap}).  This robustness degrades rapidly 
with noise: at $\sigma_{\max}=0.5$, survival drops to 85\% after 
one step of latency and to 22\% by age~3.  The reward surface 
(Fig.~\ref{fig:reward_surface}) is stable along either axis 
(low noise \emph{or} low age) but collapses when both are elevated.  
This asymmetry motivates state-dependent scheduling: transmissions 
can be deferred during benign conditions but become urgent near the 
degradation boundary.

\begin{figure}[htbp]
    \centering
    \begin{subfigure}[b]{\columnwidth}
        \includegraphics[width=\textwidth]{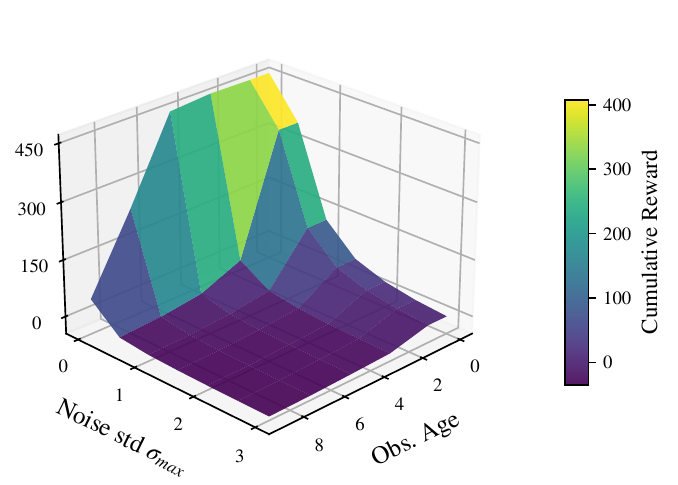}
        \caption{Reward surface}
        \label{fig:reward_surface}
    \end{subfigure}
    \hfill
    \begin{subfigure}[b]{\columnwidth}
        \includegraphics[width=\textwidth]{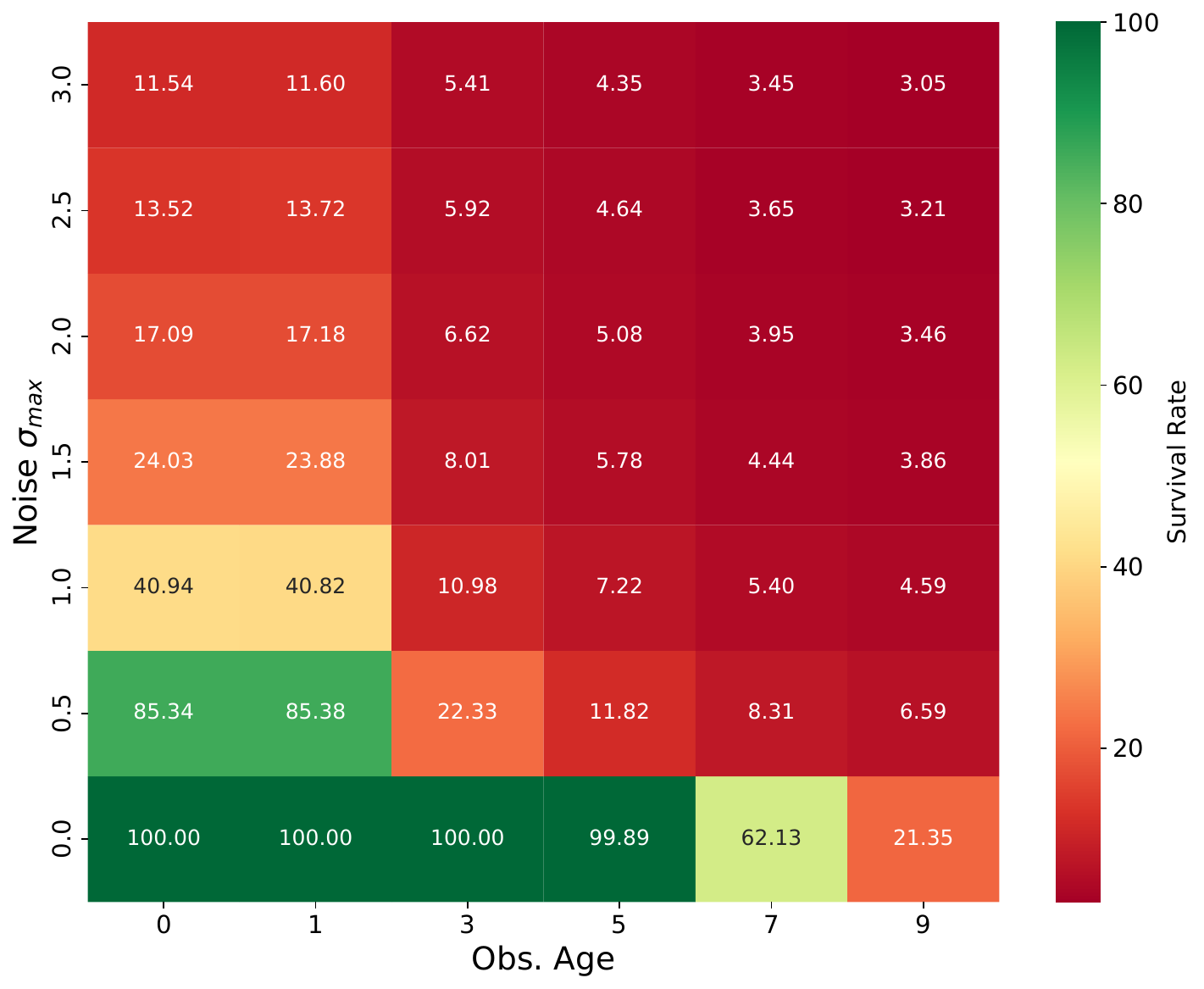}
        \caption{Survival rate (\%)}
        \label{fig:survival_heatmap}
    \end{subfigure}
    \caption{Impact of observation age and noise $\sigma_{\max}$ on 
    controller performance.}
    \label{fig:combined_performance}
\end{figure}

\paragraph{Expiration time.}
Let $\{\mathbf{x}^{(k)}\}_{k=0}^{K}$ be the \emph{stale-observation 
trajectory} obtained when the controller reuses observation 
$\mathbf{y}_n$ indefinitely:
\begin{equation}\label{eq:stale-trajectory}
 \mathbf{x}^{(k+1)} = f(\mathbf{x}^{(k)}, u^{(k)}, w^{(k)}).
\end{equation}
The \emph{expiration time} is the longest reuse horizon before the 
reward drops by more than~$\epsilon$:
\begin{equation}\label{eq:tau-star}
    T(\mathbf{y}_n; \epsilon) = \max\Bigl\{ k \in \{0, \ldots, K\} : r_0 - r_k \le \epsilon \Bigr\},
\end{equation}
where $r_k = r(\mathbf{x}^{(k+1)}, t_{n+k})$.  The expiration time 
depends on two factors: the system state (which determines proximity 
to instability boundaries) and the measurement noise (which governs 
observation reliability).  A feedforward 
network trained on rollout data predicts $T(\mathbf{y}_n;\epsilon)$ 
at deployment (details in Appendix~\ref{app:cartpole-details}).

\begin{remark}[On the tolerance parameters]
\label{rem:not-artificial}
Unlike AoI, which depends only on timestamps, the expiration time in 
\eqref{eq:tau-star} (and its linear counterpart \eqref{eq:expiration-time}) 
involves a tolerance parameter ($\epsilon$ or $\tau_n$).  These are 
not free design knobs: $\tau_n$ is set by the controller's 
maximum acceptable MSE, and $\epsilon$ by the reward loss the 
CartPole system can sustain before diverging.  Both are 
determined by the plant dynamics and the control objective.  
The expiration time is therefore grounded in the physics of 
the plant: it measures how long a sample remains useful 
\emph{for the task at hand}, not merely ``fresh'' in the 
AoI sense.
\end{remark}

\begin{figure}[htbp]
    \centering
    \begin{subfigure}[b]{\columnwidth}
        \includegraphics[width=.9\textwidth]{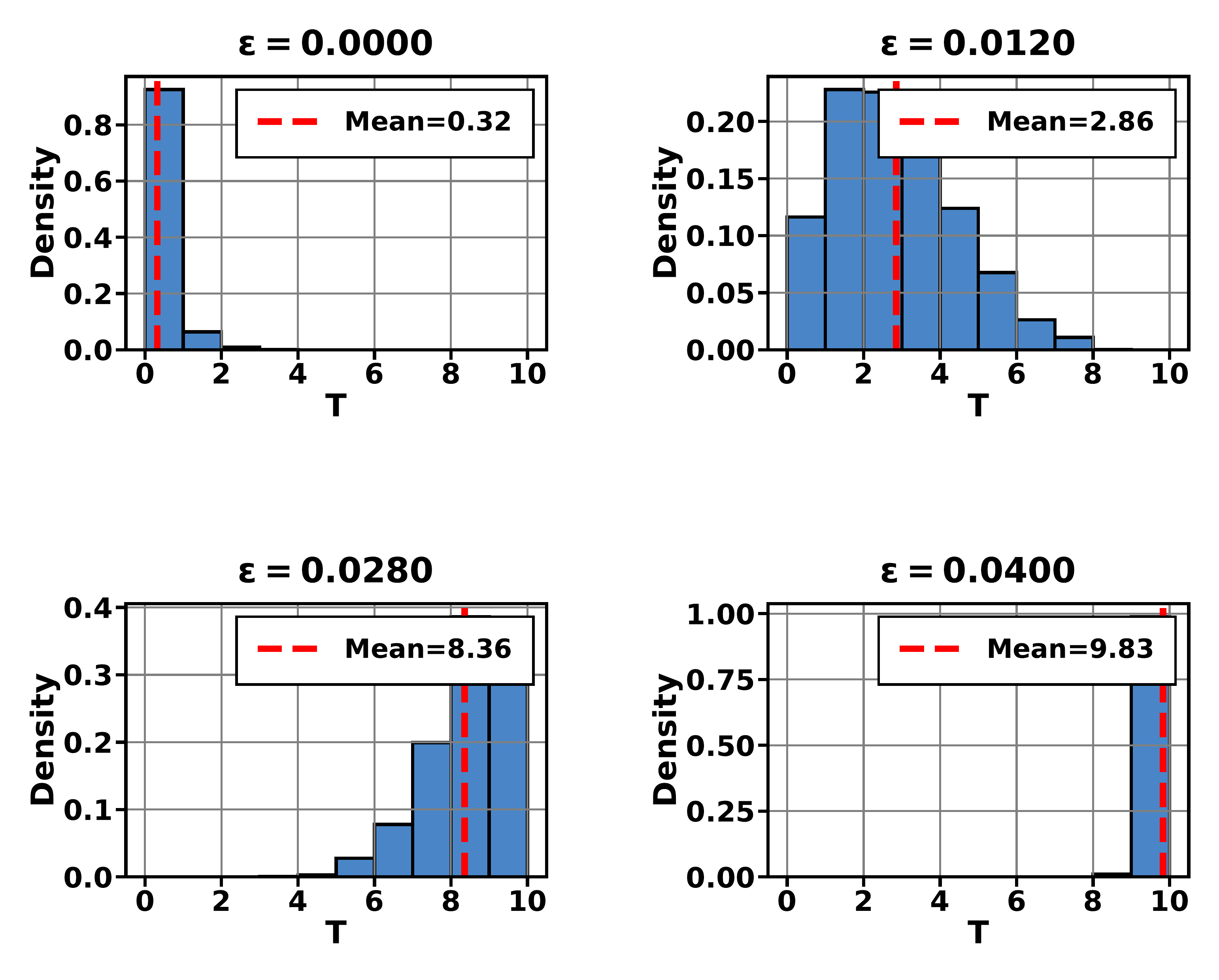}
        \caption{Distribution of $T$ for varying $\epsilon$}
        \label{fig:histograms}
    \end{subfigure}
    \hfill
    \begin{subfigure}[b]{.9\columnwidth}
        \includegraphics[width=\textwidth]{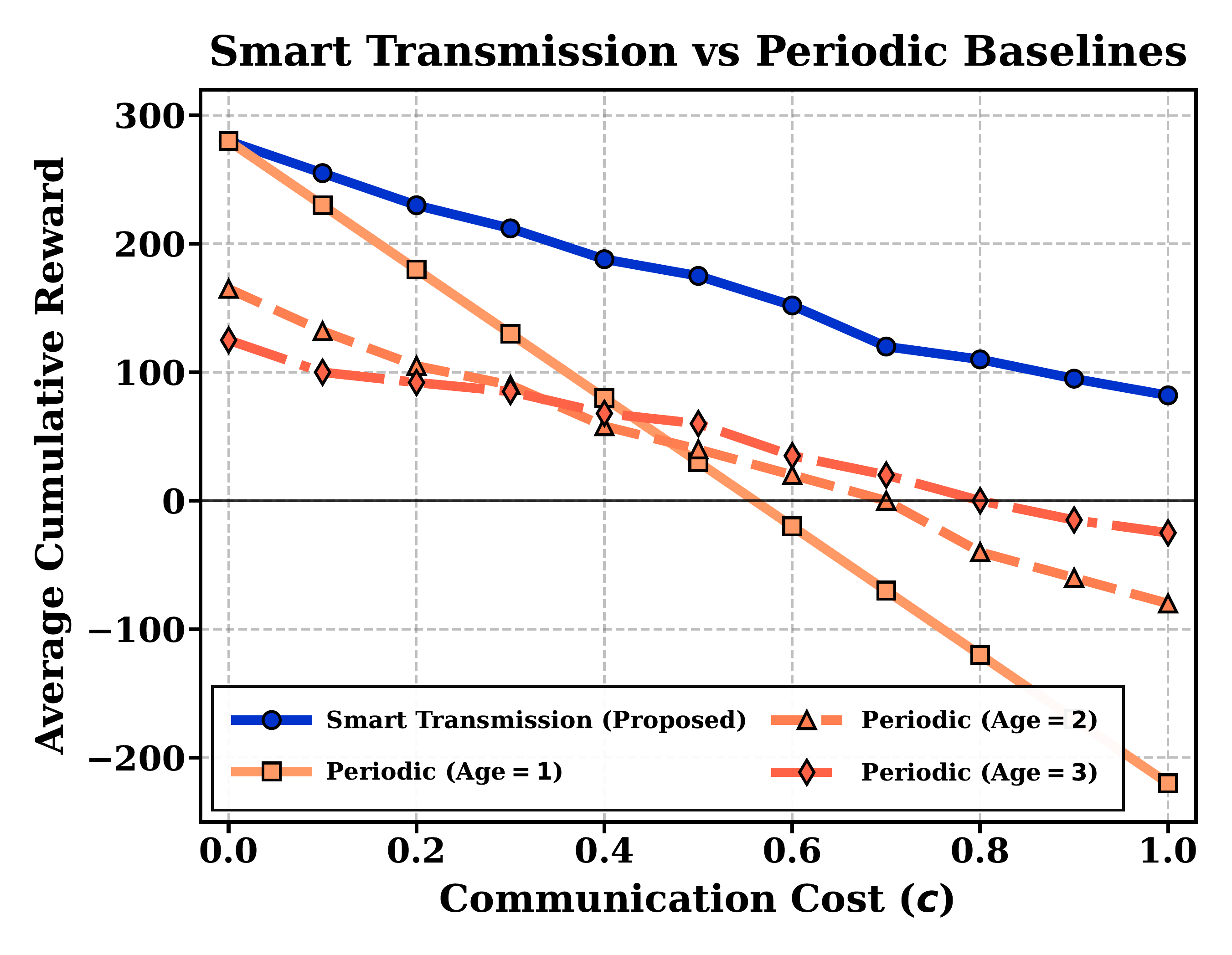}
        \caption{Reward vs.\ communication cost}
        \label{fig:cr_plot}
    \end{subfigure}
    \caption{\textbf{Expiration Analysis and Scheduling Performance.}}
    \label{fig:results_analysis}
\end{figure}

Fig.~\ref{fig:histograms} shows that the expiration time is 
state-dependent: even for a fixed~$\epsilon$, it ranges from 
$T = 1$ to $T > 10$ depending on the instantaneous state and 
noise level.  Increasing~$\epsilon$ shifts the distribution 
rightward, trading tracking precision for fewer transmissions.

At each step the scheduler compares the predicted lifetime 
$\hat{T}_n$ of the new observation against the remaining 
lifetime $T^{\mathrm{rem}}_{n-1}$ of the currently held 
sample and keeps whichever lasts longer.  Transmission occurs 
when the selected sample expires.  For each cost~$c$, the 
tolerance~$\epsilon$ is chosen to maximize cumulative reward.
Fig.~\ref{fig:cr_plot} shows that the proposed method achieves 
nearly $50\%$ higher reward than periodic baselines at moderate 
costs ($c \approx 0.4$): the scheduler defers transmissions 
when lifetimes are long and transmits promptly when noise or 
maneuvering shortens validity.
This gap follows from the structure in 
Fig.~\ref{fig:survival_heatmap}: the dynamic scheduler avoids 
redundant transmissions during benign conditions, reserving budget 
for moments when noise or maneuvering threaten stability.  Periodic 
policies, by contrast, waste resources during safe intervals while 
under-communicating during high-risk phases.

We have characterized observation expiration both analytically 
and numerically via learned predictors.
\paragraph{Common structure.}
In both case studies the expiration time measures the same thing: how 
many slots a sample can be reused before the controller's performance 
drops below a prescribed tolerance.  The tolerance is $\tau_n$ (MSE 
cap) in the linear case and $\epsilon$ (reward degradation bound) 
in CartPole; the expiration time is computed from 
\eqref{eq:expiration-time} and \eqref{eq:tau-star}, respectively.  
The coupon-collector MDP formulated next takes only the resulting 
distribution~$p_T$ as input, so its structural results 
(Sections~\ref{sec:MDP}--\ref{sec:structural}) apply to any plant 
that produces such a distribution.

Under reliable communication, the 
optimal policy is simple: transmit \emph{just-in-time}, precisely when 
the current observation expires. This holds regardless of the incoming 
sample's expiration time.

Unreliable channels break this structure. When transmissions may fail, 
the scheduler faces a gambling problem: should it send now, risking 
packet loss, or wait for a potentially better sample? The \emph{distribution} 
of expiration times and the \emph{value} of the current sample now enter 
the decision. In the next section, we formalize this 
as a coupon collector problem, where a sender holds samples of heterogeneous 
value and must decide when to transmit across an unreliable channel to 
a receiver who requires timely information.

%===============================================================
\section{Coupon-Collector View and MDP Formulation}
Having characterized expiration times, we now formulate the unreliable-channel scheduling problem as an MDP. We adopt a coupon-collector perspective where samples are expiring coupons shipped over a lossy channel, then derive the Bellman equation governing optimal transmission decisions.
\label{sec:MDP}

Note that the expiration time is computed entirely at the sensor: 
$T^{\mathrm{MSE}}_m$ (Definition~\ref{def:expiration}) depends on 
the Kalman prediction covariance, which evolves independently of 
whether the packet reaches the controller.  Similarly, 
$T(\mathbf{y}_n;\epsilon)$ in CartPole is evaluated by rolling out 
the stale-observation trajectory at the sensor.  The channel 
success probability~$p_s$ therefore enters the MDP only through 
the transition of the receiver timer~$T_r$, keeping the state 
space two-dimensional.
%===============================================================
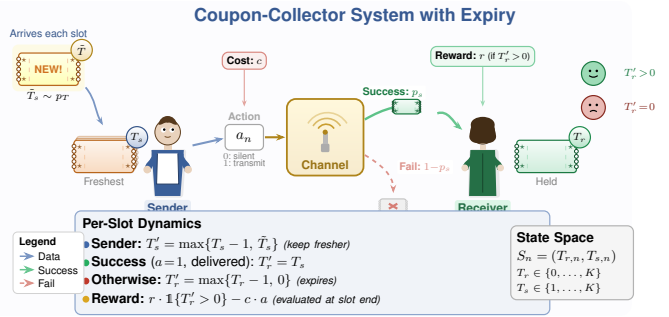
\begin{figure}[t]
    \centering
    % \resizebox forces the tikzpicture to fit exactly within the column width
    \resizebox{\columnwidth}{!}{%
    \begin{tikzpicture}[
        % Removed fixed scale=0.78 so resizebox handles the sizing
        >=Stealth,
        % Timer badge
        timer/.style={
            circle,
            draw=#1!80!black,
            thick,
            top color=white,
            bottom color=#1!35,
            minimum size=0.5cm,
            font=\sffamily\scriptsize\bfseries, % Increased font
            inner sep=0.5pt
        },
        timer/.default=senderblue,
        % Action box
        actionbox/.style={
            draw=gray!70, 
            thick, 
            fill=white, 
            rounded corners=3pt,
            minimum width=0.9cm, 
            minimum height=0.6cm
        },
        % Info box
        infobox/.style={
            draw=#1!60, 
            thick, 
            rounded corners=4pt, 
            fill=#1!8,
            inner sep=5pt
        },
        infobox/.default=gray,
        % Flow arrow style
        flowarrow/.style={
            -{Stealth[length=2.2mm, width=1.6mm]},
            line width=1.4pt,
            #1,
            shorten >=1pt,
            shorten <=1pt
        },
        curvedarrow/.style={
            -{Stealth[length=2mm, width=1.4mm]},
            line width=1.2pt,
            #1,
            shorten >=1pt,
            shorten <=1pt
        }
    ]
    
        % ================================================================
        %                    COUPON SHAPE DEFINITIONS
        % ================================================================
        \newcommand{\drawcoupon}[4]{% x, y, main color, secondary color
            \begin{scope}[shift={(#1,#2)}]
                % Shadow
                \fill[black, opacity=0.08, rounded corners=2pt] 
                    (-0.7+0.02,-0.38-0.02) rectangle (0.7+0.02,0.38-0.02);
                % Main body
                \fill[#3!15, rounded corners=2pt] (-0.7,-0.38) rectangle (0.7,0.38);
                \shade[top color=#3!8, bottom color=#3!28, rounded corners=2pt] 
                    (-0.7,-0.38) rectangle (0.7,0.38);
                \draw[#4!80!black, line width=0.8pt, rounded corners=2pt] 
                    (-0.7,-0.38) rectangle (0.7,0.38);
                % Scalloped edges
                \foreach \y in {-0.24,-0.12,0,0.12,0.24} {
                    \fill[bglight] (-0.7,\y) circle (0.045);
                    \draw[#4!60!black, line width=0.4pt] (-0.7,\y) circle (0.045);
                    \fill[bglight] (0.7,\y) circle (0.045);
                    \draw[#4!60!black, line width=0.4pt] (0.7,\y) circle (0.045);
                }
                % Dashed lines
                \draw[#4!40, line width=0.5pt, dashed] (-0.45,-0.28) -- (-0.45,0.28);
                \draw[#4!40, line width=0.5pt, dashed] (0.45,-0.28) -- (0.45,0.28);
                % Stars
                \node[font=\fontsize{6}{6}\selectfont, text=#4!60!black] at (-0.57,0.22) {$\star$};
                \node[font=\fontsize{6}{6}\selectfont, text=#4!60!black] at (0.57,0.22) {$\star$};
                \node[font=\fontsize{6}{6}\selectfont, text=#4!60!black] at (-0.57,-0.22) {$\star$};
                \node[font=\fontsize{6}{6}\selectfont, text=#4!60!black] at (0.57,-0.22) {$\star$};
            \end{scope}
        }
         
        \newcommand{\drawsuccesscoupon}[2]{% x, y
            \begin{scope}[shift={(#1,#2)}]
                \fill[black, opacity=0.08, rounded corners=2pt] 
                    (-0.7+0.02,-0.38-0.02) rectangle (0.7+0.02,0.38-0.02);
                \fill[successgreen!12, rounded corners=2pt] (-0.7,-0.38) rectangle (0.7,0.38);
                \shade[top color=successgreen!5, bottom color=successgreen!22, rounded corners=2pt] 
                    (-0.7,-0.38) rectangle (0.7,0.38);
                \draw[successgreen!70!black, line width=0.8pt, rounded corners=2pt] 
                    (-0.7,-0.38) rectangle (0.7,0.38);
                \foreach \y in {-0.24,-0.12,0,0.12,0.24} {
                    \fill[bglight] (-0.7,\y) circle (0.045);
                    \draw[successgreen!50!black, line width=0.4pt] (-0.7,\y) circle (0.045);
                    \fill[bglight] (0.7,\y) circle (0.045);
                    \draw[successgreen!50!black, line width=0.4pt] (0.7,\y) circle (0.045);
                }
                \draw[successgreen!35, line width=0.5pt, dashed] (-0.45,-0.28) -- (-0.45,0.28);
                \draw[successgreen!35, line width=0.5pt, dashed] (0.45,-0.28) -- (0.45,0.28);
                \node[font=\fontsize{6}{6}\selectfont, text=successgreen!60!black] at (-0.57,0.22) {$\star$};
                \node[font=\fontsize{6}{6}\selectfont, text=successgreen!60!black] at (0.57,0.22) {$\star$};
                \node[font=\fontsize{6}{6}\selectfont, text=successgreen!60!black] at (-0.57,-0.22) {$\star$};
                \node[font=\fontsize{6}{6}\selectfont, text=successgreen!60!black] at (0.57,-0.22) {$\star$};
            \end{scope}
        }
         
        \newcommand{\drawexpiredcoupon}[2]{% x, y
            \begin{scope}[shift={(#1,#2)}]
                \fill[gray!12, rounded corners=2pt] (-0.7,-0.38) rectangle (0.7,0.38);
                \draw[gray!45, line width=0.8pt, rounded corners=2pt] 
                    (-0.7,-0.38) rectangle (0.7,0.38);
                \foreach \y in {-0.24,-0.12,0,0.12,0.24} {
                    \fill[bglight] (-0.7,\y) circle (0.045);
                    \draw[gray!35, line width=0.4pt] (-0.7,\y) circle (0.045);
                    \fill[bglight] (0.7,\y) circle (0.045);
                    \draw[gray!35, line width=0.4pt] (0.7,\y) circle (0.045);
                }
                % X mark
                \draw[dangered!65, line width=1.5pt] (-0.22,-0.13) -- (0.22,0.13);
                \draw[dangered!65, line width=1.5pt] (-0.22,0.13) -- (0.22,-0.13);
            \end{scope}
        }
    
        % Background
        \begin{scope}[on background layer]
            \fill[bglight, rounded corners=8pt] (-3.6,-4.2) rectangle (11.5,3.2);
        \end{scope}
    
        % ================================================================
        %                          TITLE
        % ================================================================
        \node[font=\sffamily\large\bfseries, text=senderblue!85!black] 
            at (5, 2.85) {Coupon-Collector System with Expiry};
    
        % ================================================================
        %                    FRESH COUPON ARRIVAL (TOP LEFT)
        % ================================================================
        \begin{scope}[shift={(-2.2, 1.6)}]
            % Glow
            \fill[yellow!35, opacity=0.35, rounded corners=4pt] 
                (-0.85,-0.52) rectangle (0.85,0.58);
            % Coupon
            \drawcoupon{0}{0}{yellow}{couponorange}
            % NEW label
            \node[font=\sffamily\fontsize{7}{7}\selectfont\bfseries, text=couponorange!90!black] at (0,0) {NEW!};
            % Timer
            \node[timer=channelgold] at (0.78,0.45) {$\tilde{T}$};
            % Labels
            \node[font=\sffamily\fontsize{7}{7}\selectfont, text=senderblue!80!black, align=center] 
                at (0,0.85) {Arrives each slot};
            \node[font=\sffamily\fontsize{7}{7}\selectfont] at (0,-0.6) {$\tilde{T}_s \sim p_T$};
        \end{scope}
    
        % Arrow: New coupon to Sender stack
        \draw[curvedarrow=senderblue!70] 
            (-1.35, 1.2) to[out=-50, in=120] (-0.9, 0.25);
    
        % ================================================================
        %                          SENDER (SHOPKEEPER)
        % ================================================================
        \begin{scope}[shift={(0.6,0)}]
            % Shadow
            \fill[black, opacity=0.08, rounded corners=2pt] 
                (-0.32,-1.42) -- (-0.40,-0.28) -- (0.44,-0.28) -- (0.36,-1.42) -- cycle;
            % Body
            \fill[senderblue!90!black] (-0.3,-1.38) -- (-0.38,-0.3) -- (0.38,-0.3) -- (0.3,-1.38) -- cycle;
            % Apron
            \fill[white] (-0.24,-0.38) -- (-0.26,-1.1) -- (0.26,-1.1) -- (0.24,-0.38) -- cycle;
            \draw[senderblue!60!black, line width=0.6pt] 
                (-0.24,-0.38) -- (-0.26,-1.1) -- (0.26,-1.1) -- (0.24,-0.38);
            % Neck
            \fill[brown!35!white] (-0.09,-0.19) rectangle (0.09,-0.3);
            % Head
            \fill[brown!35!white] (0,0.12) circle (0.3);
            \draw[brown!50!black, line width=0.6pt] (0,0.12) circle (0.3);
            % Eyes
            \fill[white] (-0.1,0.17) circle (0.06);
            \fill[white] (0.1,0.17) circle (0.06);
            \fill[black] (-0.08,0.18) circle (0.03);
            \fill[black] (0.12,0.18) circle (0.03);
            % Smile
            \draw[brown!60!black, line width=0.6pt] (-0.07,0.02) to[out=-30,in=210] (0.07,0.02);
            % Hair
            \fill[brown!45!black] 
                (-0.22,0.3) to[out=80,in=180] (0,0.5) 
                to[out=0,in=100] (0.22,0.3) 
                to[out=200,in=20] (0,0.34) 
                to[out=160,in=-20] cycle;
            % Arms
            \fill[senderblue!90!black] 
                (-0.38,-0.38) to[out=-70,in=120] (-0.52,-0.76) 
                to[out=-30,in=-150] (-0.4,-0.8) 
                to[out=80,in=-60] (-0.3,-0.45) -- cycle;
            \fill[senderblue!90!black] 
                (0.38,-0.38) to[out=-110,in=60] (0.52,-0.76) 
                to[out=-150,in=-30] (0.4,-0.8) 
                to[out=100,in=-120] (0.3,-0.45) -- cycle;
            % Hands
            \fill[brown!35!white] (-0.48,-0.76) circle (0.09);
            \fill[brown!35!white] (0.48,-0.76) circle (0.09);
            % Label
            \node[fill=senderblue!12, rounded corners=2pt, inner sep=2pt,
                  font=\sffamily\fontsize{8}{8}\selectfont\bfseries, text=senderblue!80!black] 
                at (0,-1.65) {Sender};
        \end{scope}
    
        % ================================================================
        %                    COUPONS AT SENDER
        % ================================================================
        \begin{scope}[shift={(-0.9,-0.45)}]
            % Stack effect
            \begin{scope}[opacity=0.35, shift={(0.06, 0.09)}]
                \drawcoupon{0}{0}{couponorange}{couponorange}
            \end{scope}
            \begin{scope}[opacity=0.55, shift={(0.03, 0.045)}]
                \drawcoupon{0}{0}{couponorange}{couponorange}
            \end{scope}
            % Front coupon
            \drawcoupon{0}{0}{couponorange}{couponorange}
            % Timer
            \node[timer=senderblue] at (0.78,0.45) {$T_s$};
            % Label
            \node[font=\sffamily\fontsize{7}{7}\selectfont, text=gray, align=center] at (0,-0.58) {Freshest};
        \end{scope}
    
        % ================================================================
        %                       ACTION DECISION
        % ================================================================
        \node[actionbox] (action) at (2.4, 0) {};
        \node[font=\sffamily\normalsize\bfseries] at (2.4,0) {$a_n$};
        \node[font=\sffamily\fontsize{7}{7}\selectfont\bfseries, text=gray!70] at (2.4, 0.5) {Action};
        \node[font=\sffamily\fontsize{6}{6}\selectfont, text=gray, align=left] at (2.4, -0.52) {
            $0$: silent\\[-1pt]
            $1$: transmit
        };
    
        % Arrow: Sender to Action
        \draw[flowarrow=senderblue!70] (1.15, -0.2) -- (1.95, -0.02);
    
        % ================================================================
        %                          CHANNEL
        % ================================================================
        \begin{scope}[shift={(4.3,0)}]
            % Box
            \shade[top color=channellight, bottom color=channelgold!22, rounded corners=6pt] 
                (-0.9,-0.9) rectangle (0.9,0.9);
            \draw[channelgold!80!black, line width=1.2pt, rounded corners=6pt] 
                (-0.9,-0.9) rectangle (0.9,0.9);
            % Antenna
            \fill[gray!65, rounded corners=1pt] (-0.1,-0.35) rectangle (0.1,-0.15);
            \fill[gray!55] (-0.05,-0.15) rectangle (0.05,0.12);
            \fill[channelgold!80!black] (0,0.12) circle (0.06);
            % Waves
            \foreach \r/\o in {0.22/0.9, 0.38/0.6, 0.54/0.35} {
                \draw[channelgold!90!black, line width=1.1pt, opacity=\o] 
                    ({-\r*0.75},{0.12+\r*0.38}) arc (148:32:\r);
            }
            % Label
            \node[font=\sffamily\fontsize{8}{8}\selectfont\bfseries, text=channelgold!70!black] 
                at (0,-0.65) {Channel};
            \node[font=\sffamily\fontsize{6}{6}\selectfont, text=gray!65] at (0,-0.82) {};
        \end{scope}
    
        % Arrow: Action to Channel
        \draw[flowarrow=channelgold!80!black] (2.85, 0) -- (3.4, 0);
    
        % ================================================================
        %                       SUCCESS PATH (ABOVE)
        % ================================================================
        % Arrow curves upward to avoid overlap
        \draw[flowarrow=successgreen!75, line width=1.6pt] 
            (5.2, 0.4) to[out=30, in=150] (6.6, 0.55);
         
        % Success label
        \node[font=\sffamily\fontsize{7}{7}\selectfont\bfseries, text=successgreen!70!black,
              fill=white, inner sep=1.5pt, rounded corners=1pt] 
            at (5.9, 1.05) {Success: $p_s$};
         
        % Small success coupon indicator
        \begin{scope}[scale=0.45, shift={(13.8, 1.6)}]
            \drawsuccesscoupon{0}{0}
        \end{scope}
    
        % ================================================================
        %                       FAILURE PATH (BELOW)
        % ================================================================
        \draw[-{Stealth[length=2mm, width=1.4mm]}, 
              line width=1.2pt, dangered!55, dashed, 
              shorten >=1pt, shorten <=1pt] 
            (5.2, -0.4) to[out=-30, in=90] (5.9, -1.35);
         
        % Failure label
        \node[font=\sffamily\fontsize{7}{7}\selectfont\bfseries, text=dangered!55,
              fill=white, inner sep=1.5pt, rounded corners=1pt] 
            at (6.65, -0.7) {Fail: $1{-}p_s$};
         
        % Failed coupon
        \begin{scope}[scale=0.45]
            \drawexpiredcoupon{13.1}{-3.6}
        \end{scope}
    
        % ================================================================
        %                          RECEIVER (CUSTOMER)
        % ================================================================
        \begin{scope}[shift={(8,0)}]
            % Shadow
            \fill[black, opacity=0.08, rounded corners=2pt] 
                (-0.35,-1.42) -- (-0.31,-0.28) -- (0.37,-0.28) -- (0.41,-1.42) -- cycle;
            % Dress
            \fill[receivergreen!80!black] 
                (-0.32,-1.38) -- (-0.29,-0.3) -- (0.29,-0.3) -- (0.32,-1.38) -- cycle;
            \draw[receivergreen!50!black, line width=0.5pt] (0,-0.38) -- (0,-1.3);
            % Neck
            \fill[brown!30!white] (-0.07,-0.19) rectangle (0.07,-0.3);
            % Head
            \fill[brown!30!white] (0,0.12) circle (0.28);
            \draw[brown!45!black, line width=0.6pt] (0,0.12) circle (0.28);
            % Eyes
            \fill[white] (-0.1,0.17) circle (0.055);
            \fill[white] (0.1,0.17) circle (0.055);
            \fill[black] (-0.08,0.18) circle (0.03);
            \fill[black] (0.12,0.18) circle (0.03);
            % Smile
            \draw[brown!50!black, line width=0.6pt] (-0.07,0.0) to[out=-35,in=215] (0.07,0.0);
            % Hair
            \fill[brown!35!black] 
                (-0.28,0.28) to[out=95,in=180] (0,0.48) 
                to[out=0,in=85] (0.28,0.28)
                to[out=260,in=80] (0.31,-0.02)
                to[out=250,in=0] (0.07,-0.14)
                to[out=180,in=290] (-0.31,-0.02)
                to[out=100,in=280] cycle;
            % Arms
            \fill[receivergreen!80!black] 
                (-0.29,-0.38) to[out=-70,in=120] (-0.45,-0.72) 
                to[out=-30,in=-150] (-0.34,-0.76) 
                to[out=80,in=-60] (-0.24,-0.45) -- cycle;
            \fill[receivergreen!80!black] 
                (0.29,-0.38) to[out=-110,in=60] (0.45,-0.72) 
                to[out=-150,in=-30] (0.34,-0.76) 
                to[out=100,in=-120] (0.24,-0.45) -- cycle;
            % Hands
            \fill[brown!30!white] (-0.42,-0.72) circle (0.08);
            \fill[brown!30!white] (0.42,-0.72) circle (0.08);
            % Label
            \node[fill=receiverlight, rounded corners=2pt, inner sep=2pt,
                  font=\sffamily\fontsize{8}{8}\selectfont\bfseries, text=receivergreen!80!black] 
                at (0,-1.65) {Receiver};
        \end{scope}
    
        % Arrow: Success to Receiver
        \draw[flowarrow=successgreen!70, line width=1.6pt] 
            (6.9, 0.55) to[out=20, in=140] (7.55, 0.15);
    
        % ================================================================
        %                    COUPON AT RECEIVER
        % ================================================================
        \begin{scope}[shift={(9.5,-0.45)}]
            \drawsuccesscoupon{0}{0}
            % Timer
            \node[timer=successgreen] at (0.78,0.45) {$T_r$};
            % Label
            \node[font=\sffamily\fontsize{7}{7}\selectfont, text=gray, align=center] at (0,-0.58) {Held};
        \end{scope}
    
        % ================================================================
        %                    COST & REWARD
        % ================================================================
        % Cost
        \node[infobox=dangered, font=\sffamily\fontsize{7}{7}\selectfont, inner sep=3pt] 
            (cost_box) at (2.4, 1.8) {\textbf{Cost:} $c$};
        \draw[thick, dangered!45, -{Stealth[length=1.5mm]}, shorten >=2pt] 
            (cost_box.south) -- (2.4, 0.65);
    
        % Reward
        \node[infobox=successgreen, font=\sffamily\fontsize{7}{7}\selectfont, align=center, inner sep=3pt] 
            (reward_box) at (8, 1.9) {\textbf{Reward:} $r$ {\fontsize{6}{6}\selectfont(if $T'_r\!>\!0$)}};
        \draw[thick, successgreen!45, -{Stealth[length=1.5mm]}, shorten >=2pt] 
            (reward_box.south) -- (8, 0.65);
    
        % ================================================================
        %                    HAPPY / SAD STATES
        % ================================================================
        \begin{scope}[shift={(10.6, 1.5)}]
            % Happy
            \fill[successgreen!45] (0,0) circle (0.28);
            \draw[successgreen!70!black, line width=0.8pt] (0,0) circle (0.28);
            \fill[black] (-0.08,0.06) circle (0.03);
            \fill[black] (0.08,0.06) circle (0.03);
            \draw[black, line width=0.8pt] (-0.1,-0.05) to[out=-40,in=220] (0.1,-0.05);
            \node[font=\sffamily\fontsize{7}{7}\selectfont, right=0.32cm, text=successgreen!70!black] 
                at (0.28,0) {$T'_r\!>\!0$};
        \end{scope}
    
        \begin{scope}[shift={(10.6, 0.7)}]
            % Sad
            \fill[dangered!35] (0,0) circle (0.28);
            \draw[dangered!70!black, line width=0.8pt] (0,0) circle (0.28);
            \fill[black] (-0.08,0.06) circle (0.03);
            \fill[black] (0.08,0.06) circle (0.03);
            \draw[black, line width=0.8pt] (-0.08,-0.1) to[out=40,in=140] (0.08,-0.1);
            \node[font=\sffamily\fontsize{7}{7}\selectfont, right=0.32cm, text=dangered!70!black] 
                at (0.28,0) {$T'_r\!=\!0$};
        \end{scope}
    
        % ================================================================
        %                    DYNAMICS BOX (COMPACT)
        % ================================================================
        \node[infobox=senderblue, text width=10.8cm, align=left, 
              font=\sffamily\fontsize{9}{10}\selectfont, inner sep=5pt] 
            at (4, -2.95) {
            \textbf{Per-Slot Dynamics}\\[3pt]
            \begin{tabular}{@{}c@{\,}l@{}}
                \tikz\fill[senderblue] (0,0) circle (0.08); & 
                \textbf{Sender:} 
                $T'_{s} = \max\{T_{s} - 1,\, \tilde{T}_{s}\}$ 
                \textit{\fontsize{7}{7}\selectfont(keep fresher)}\\[2pt]
                \tikz\fill[successgreen] (0,0) circle (0.08); & 
                \textbf{Success} ($a\!=\!1$, delivered): 
                $T'_{r} = T_{s}$\\[2pt]
                \tikz\fill[dangered] (0,0) circle (0.08); & 
                \textbf{Otherwise:} 
                $T'_{r} = \max\{T_{r} - 1,\, 0\}$ 
                \textit{\fontsize{7}{7}\selectfont(expires)}\\[2pt]
                \tikz\fill[channelgold] (0,0) circle (0.08); & 
                \textbf{Reward:} 
                $r \cdot \mathds{1}\{T'_r > 0\} - c \cdot a$
                \textit{\fontsize{7}{7}\selectfont(evaluated at slot end)}\\
            \end{tabular}
        };
    
        % ================================================================
        %                    STATE SPACE BOX (BOTTOM RIGHT)
        % ================================================================
        \node[infobox=gray, text width=2.6cm, align=left, 
              font=\sffamily\fontsize{9}{10}\selectfont, inner sep=4pt] 
            at (10.1, -2.95) {
            \textbf{State Space}\\[2pt]
            $S_n = (T_{r,n}, T_{s,n})$\\[1pt]
            {\fontsize{7}{7}\selectfont $T_r \in \{0,\ldots,K\}$}\\
            {\fontsize{7}{7}\selectfont $T_s \in \{1,\ldots,K\}$}
        };
    
        % ================================================================
        %                    LEGEND (BOTTOM LEFT - NO OVERLAP)
        % ================================================================
        \node[draw=gray!35, rounded corners=3pt, fill=white, inner sep=4pt, 
              font=\sffamily\fontsize{7}{8}\selectfont, align=left] 
            at (-2.2, -2.95) {
            \textbf{Legend}\\[2pt]
            \tikz\draw[-{Stealth}, senderblue!70, line width=1pt] (0,0) -- (0.35,0); Data\\[1pt]
            \tikz\draw[-{Stealth}, successgreen!70, line width=1pt] (0,0) -- (0.35,0); Success\\[1pt]
            \tikz\draw[-{Stealth}, dangered!55, dashed, line width=1pt] (0,0) -- (0.35,0); Fail
        };
    
    \end{tikzpicture}
    } % End of resizebox
    \caption{Coupon-collector model with expiring samples.}
    \label{fig:coupon_collector}
\end{figure}

\subsection{A Coupon-Collector Variant with Expiry}

Consider a shopkeeper receiving one coupon per day, each valid for a random number of days $T \in \{1,\dots,K\}$ drawn i.i.d.\ from $p_T$. The shopkeeper may ship the coupon to a customer incurring a fee $c$ with delivery failure probability $1-p_s$, or may choose to keep it (possibly expiring unused). The customer earns reward $r$ for each day she holds an unexpired coupon. The goal is to maximize average reward minus shipping fees.

Unlike classical coupon-collector problems that minimize collection time with non-expiring coupons~\cite{ross2010introduction}, our variant features strict expiration times and a steady-state reward--cost trade-off, requiring new analytical tools.

\subsection{MDP Components}
\label{ssec:MDP-components}
We now formalize the coupon-collector analogy as a Markov decision process, with components illustrated in Figure~\ref{fig:coupon_collector}. The shopkeeper maps to the sender, the customer to the receiver, and coupon validity to sample expiration time.
\paragraph{State and Action.}
Let $T_{r,n} \in \{0,\dots,K\}$ denote the receiver's sample remaining lifetime ($T_{r,n}=0$ means no valid sample), and $T_{s,n} \in \{1,\dots,K\}$ the sender's freshest sample lifetime. The state is $S_n = (T_{r,n}, T_{s,n})$. At each slot, the sender chooses $a_n \in \{0,1\}$ (transmit or remain silent).

\paragraph{Transitions.}
The sender keeps the fresher of the aged sample and new arrival:
$T_{s,n+1} = \max\{T_{s,n}-1, \tilde{T}_{s,n+1}\}$, where $\tilde{T}_{s,n+1} \sim p_T$.
The receiver's timer evolves as:
\begin{equation}
T_{r, n+1} = 
\begin{cases}
T_{s,n}, & \parbox[t]{.35\columnwidth}{if $a_n = 1$ and \\ delivery succeeds,} \\[1em]
\max\{T_{r,n} - 1, 0\}, & \text{otherwise}.
\end{cases}
\end{equation}

\paragraph{Reward and Objective.}
The instantaneous reward is $R(S_n, a_n, S_{n+1}) = r \cdot \mathds{1}\{T_{r,n+1} > 0\} - c \cdot a_n$. We seek a stationary policy $\pi: \mathcal{S} \to \{0,1\}$ maximizing the long-run average reward:
\[
\max_{\pi}\; \liminf_{N\to\infty}\frac{1}{N}\sum_{n=1}^{N}\mathbb{E}^\pi\bigl[R(S_n,a_n, S_{n+1})\bigr].
\]

\paragraph{Bellman Equation.}
The optimal policy satisfies the average-reward Bellman equation~\cite{bertsekas2011dynamic}:
\begin{equation} \label{eq:avg_bellman}
\rho^* + V(T_r, T_s) = \max_{a \in \{0,1\}} \bigl\{\bar{R}(S, a) + \mathbb{E}[V(T'_r, T'_s) \mid S, a]\bigr\},
\end{equation}
where $\rho^*$ is the optimal average reward (gain), $V(\cdot)$ is the bias function, and $\bar{R}(S, a) = \mathbb{E}[R(S, a, S')]$. Value iteration computes $\pi^*$ via:
\begin{equation}
\begin{split}
Q^{(k+1)}(T_r, T_s, a) = -c \cdot a + \sum_{T'_r, T'_s} P_a(S \to S') \\
\times \bigl[r \cdot \mathds{1}\{T'_r > 0\} + V^{(k)}(T'_r, T'_s)\bigr],
\end{split}
\label{eq:bellman_Q_update}
\end{equation}
but requires exact knowledge of $p_s$ and $p_T$.
The MDP admits a compact state space of  $(K+1) \times K$, making value iteration tractable when the $p_s$ and $p_T$ are known. Two questions remain: (i) does the optimal policy possess exploitable structure, and (ii) can we learn effective policies when these parameters are unknown? The next section addresses both.

\section{Optimal Scheduling Policy}
\label{sec:structural}

We now study the optimal policy from three angles. We start by establishing structural properties that hold for any expiration time distribution: the value function is monotone and the optimal policy follows a threshold rule. When expiration times are constant and parameters are known, these results yield a closed-form optimal threshold. When parameters are unknown, we show how to embed the same structural insights into a Q-learning algorithm that learns efficiently.
\begin{lemma}
\label{lem:monotone}
Let $V(T_r,T_s)$ satisfy~\eqref{eq:avg_bellman}. Then $V(T_r,T_s) \le V(T_r{+}1,T_s)$ and $V(T_r,T_s) \le V(T_r,T_s{+}1)$ for all valid states.
\end{lemma}
\begin{proof}
See Appendix~\ref{app:proof-monotone}.
\end{proof}
The monotonicity reflects natural intuitions: a receiver with more remaining validity ($T_r{+}1$ vs.\ $T_r$) can wait longer before 
requiring an update, while a sender with a fresher sample ($T_s{+}1$ vs.\ $T_s$) offers greater potential value upon 
successful transmission. Both advantages propagate through the Bellman recursion, yielding higher long-run reward.

For binary actions, the optimal policy satisfies 
$\pi^*(T_r,T_s)=\mathds{1}\{\Delta Q(T_r,T_s)>0\}$, where 
$\Delta Q = Q(\cdot,1) - Q(\cdot,0)$ is the advantage of transmitting. The following theorem shows this advantage has sufficient structure 
to yield a threshold policy.
%\begin{lemma}
%\label{lemma:monotonicity_submodularity}
%$\Delta Q(T_r, T_s)$ is non-increasing in $T_r$ and non-decreasing in $T_s$.
%\end{lemma}
%\begin{proof}
%See Appendix~\ref{app:proof-submodularity}.
%\end{proof}

\begin{theorem}
\label{thm:row_threshold}
For each $T_s \in \{1,\dots,K\}$, there exists a threshold $\theta(T_s) \in \{0,\dots,K\}$ such that $\pi^*(T_r,T_s) = 1$ iff $T_r \le \theta(T_s)$. Moreover, $\theta(T_s) \le \theta(T_s+1)$.
\end{theorem}
\begin{proof}
See Appendix~\ref{app:proof-threshold}.
\end{proof}
 If the receiver still has plenty of validity left, there is no rush to transmit. As $T_s$ increases, the potential gain from successful delivery grows, so the sender becomes more willing to attempt transmission early. The next result  identifies three properties that hold across the entire state space.

\begin{theorem}
\label{thm:global_structure}
The optimal policy $\pi^*$ satisfies:
\begin{enumerate}[label=\textnormal{(\roman*)},itemsep=1pt,leftmargin=1.4em]
\item If $T_r=0$ and $p_s r - c > 0$, then $\pi^*(0,T_s)=1$ for all $T_s$.
\item  If $T_r > T_s$, then $\pi^*(T_r,T_s)=0$.
\item  If transmitting is optimal at $(T_r, T_s)$ for $T_r > 0$, it is optimal at $(T_r-1, T_s)$.
\end{enumerate}
\begin{proof}
See Appendix~\ref{app:proof-global-structure}.
\end{proof}

\end{theorem}
Intuitively, an expired receiver should always be refreshed when transmission pays off on average. Sending data staler than what the receiver already holds is clearly wasteful since cost is incurred with no benefit. And as the receiver's validity diminishes, the urgency to transmit naturally grows.

The preceding results hold for general lifetime distributions. When 
lifetimes are constant, the threshold admits a closed-form characterization.
\begin{theorem}
\label{thm:constant_lifetime}
For constant lifetime $T_{s,n} = K$, the optimal policy is a threshold rule: transmit iff $T_r \le \theta^*$. The average reward is $\rho(\theta) = r - \frac{c + r(1-p_s)^\theta}{p_s(K-\theta) + 1}$, and $\theta^*$ is the largest integer satisfying $(K-\theta+1)(1-p_s)^{\theta-1} > c/(p_s r)$.
\end{theorem}
\begin{proof}
See Appendix~\ref{app:proof-constant-lifetime}.
\end{proof}

The threshold increases with $K$ and $r$, and decreases with $p_s$ (reliable channels require fewer transmission attempts) 
and $c$, reflecting the tradeoff between transmission cost and 
expiration risk. For random lifetimes, closed-form solutions are 
generally unavailable, motivating the learning approach below.

\subsection{Structure-Aware Q-Learning}
\label{ssec:ql_without_empty}

When the channel success probability $p_s$ and lifetime distribution $p_T$ 
are unknown, we turn to model-free reinforcement learning.  The structural 
results of Theorems~\ref{thm:row_threshold}--\ref{thm:global_structure} 
accelerate learning in three ways.  
Theorem~\ref{thm:global_structure}\,(ii) shows that transmitting 
when $T_r > T_s$ is never optimal; hard-coding $a=0$ in this region 
eliminates roughly half the state--action space.  
Theorem~\ref{thm:global_structure}\,(iii) implies the transmit region 
is downward-closed in $T_r$, which steers $\epsilon$-greedy exploration 
toward the threshold boundary even though SAQ does not parameterise a 
threshold explicitly.  Finally, for each sampled $(T_s, Z)$ pair we 
sweep all $T_r \in \{0,\ldots,K\}$, obtaining $O(K)$ Q-updates from a 
single sample.

SAQ does not enforce the full threshold structure of 
Theorem~\ref{thm:row_threshold} as a hard constraint, so a 
residual gap to optimal Value Iteration may persist, 
especially for large $K$ or heavy-tailed lifetimes.  
Parameterising the policy directly through $\theta(T_s)$ 
and learning the thresholds via stochastic approximation 
could close this gap; we leave the investigation to future work.

Algorithm~\ref{alg:q_no_empty} summarizes our approach. At each iteration, 
we sample a sender lifetime $T_s \sim p_T$, a channel realization 
$Z \sim \text{Bern}(p_s)$, and a next-slot arrival $\tilde{T}_s \sim p_T$. 
We then loop over all receiver states: for $T_r > T_s$, we force $a=0$; 
for $T_r \leq T_s$, we apply $\epsilon$-greedy exploration. Each 
$(T_r, T_s, a)$ triple receives a Q-update based on the resulting 
transition and reward.

\begin{algorithm}[!t]
\caption{Structure-Aware Q-Learning (SAQ)}
\label{alg:q_no_empty}
\begin{algorithmic}[1]
  \State \textbf{Init:} $Q_0(T_r, T_s, a) \gets 0$ for all states/actions; $\hat{\rho}_0 \gets 0$
  \For{$n = 0, 1, 2, \ldots$} 
    \State Sample $T_s \sim p_T$, $Z \sim \operatorname{Bern}(p_s)$, $\tilde{T}_s \sim p_T$
    \State $R_{\mathrm{sum}} \gets 0$ \Comment{Initialize accumulator}
    \For{$T_r = 0$ \textbf{to} $K$} \Comment{Sweep all receiver states}
      \If{$T_r > T_s$} 
        \State $a \gets 0$ \Comment{Never send obsolete}
      \Else
        \State $a \gets \epsilon\text{-greedy}\bigl(\arg\max_{a'} Q_n(T_r, T_s, a')\bigr)$
      \EndIf
      \State $T'_r \gets T_s$ if $(a=1 \land Z=1)$ else $\max\{T_r-1, 0\}$
      \State $T'_s \gets \max\{T_s - 1, \tilde{T}_s\}$
      \State $R \gets r \cdot \mathds{1}_{\{T'_r > 0\}} - c \cdot a$
      \State $R_{\mathrm{sum}} \gets R_{\mathrm{sum}} + R$ \Comment{Accumulate reward}
      \State $\delta \gets R - \hat{\rho}_n + \max_{a'} Q_n(T'_r, T'_s, a') - Q_n(T_r, T_s, a)$
      \State $Q_{n+1}(T_r, T_s, a) \gets Q_n(T_r, T_s, a) + \alpha_n \delta$
    \EndFor
    \State $\bar{R}_n \gets R_{\mathrm{sum}} / (K+1)$
    \State $\hat{\rho}_{n+1} \gets \hat{\rho}_n + \beta_n (\bar{R}_n - \hat{\rho}_n)$
  \EndFor
\end{algorithmic}
\end{algorithm}

The algorithm uses two-timescale stochastic approximation~\cite{bertsekas2011dynamic}.  Q-values evolve on step-size $\alpha_n$; the average reward estimate $\hat{\rho}$ evolves on $\beta_n = o(\alpha_n)$. With $\sum_n \alpha_n = \infty$, $\sum_n \alpha_n^2 < \infty$ (and likewise for $\beta_n$), convergence to the optimal policy is guaranteed. The structural constraints from 
Theorem~\ref{thm:global_structure} cut the exploration space roughly in half, and the batch sweep yields $O(K)$ updates per sample. Both modifications speed up learning considerably in practice.

\section{Simulation Results}
\label{sec:simulations}

We validate our structure-aware Q-learning (SAQ) on a system with maximum expiration time $K=20$, reward $r=1$, lifetime distribution $p_T$ uniform on $\{1,\dots,K\}$, and varying $p_s$ and cost ratio $c/r$.
Figure~\ref{fig:policy_heatmaps_ps} reveals how channel reliability $p_s$ shapes the optimal policy structure. At $p_s = 1.0$ (perfect channel), the send region is minimal, the scheduler waits until the receiver's sample nears expiration, confident that a single transmission will succeed. In other words, because of reliability, the optimal policy is JIT as shown by Theorem \ref{thm:optimal-policy}. As $p_s$ decreases, the send region progressively expands: at $p_s = 0.25$, the scheduler transmits earlier to hedge against repeated failures. Quantitatively, the threshold $\theta(T_s)$ increases as $p_s$ drops, confirming that unreliable channels demand more aggressive scheduling. Across all settings, the double-threshold structure from Theorem~\ref{thm:row_threshold} holds: thresholds increase monotonically with $T_s$, reflecting the value of transmitting fresher samples. The diagonal boundary $T_r = T_s$ remains inviolate as transmitting obsolete data is never optimal regardless of channel conditions.

Figure~\ref{fig:parameter_sweeps} evaluates performance of SAQ by sweeping channel reliability $p_s$ and cost ratio $c/r$. As expected, average reward increases with $p_s$ (more successful deliveries) and decreases with $c/r$ (higher transmission penalty). Across all parameter ranges, SAQ consistently stays close to the optimal $\rho^*$ from Value Iteration, while baseline Q-learning exhibits a persistent gap due to slower convergence.
\begin{figure}[t]
  \centering
  \includegraphics[width=.9\columnwidth, trim=0 80 0 60, clip]{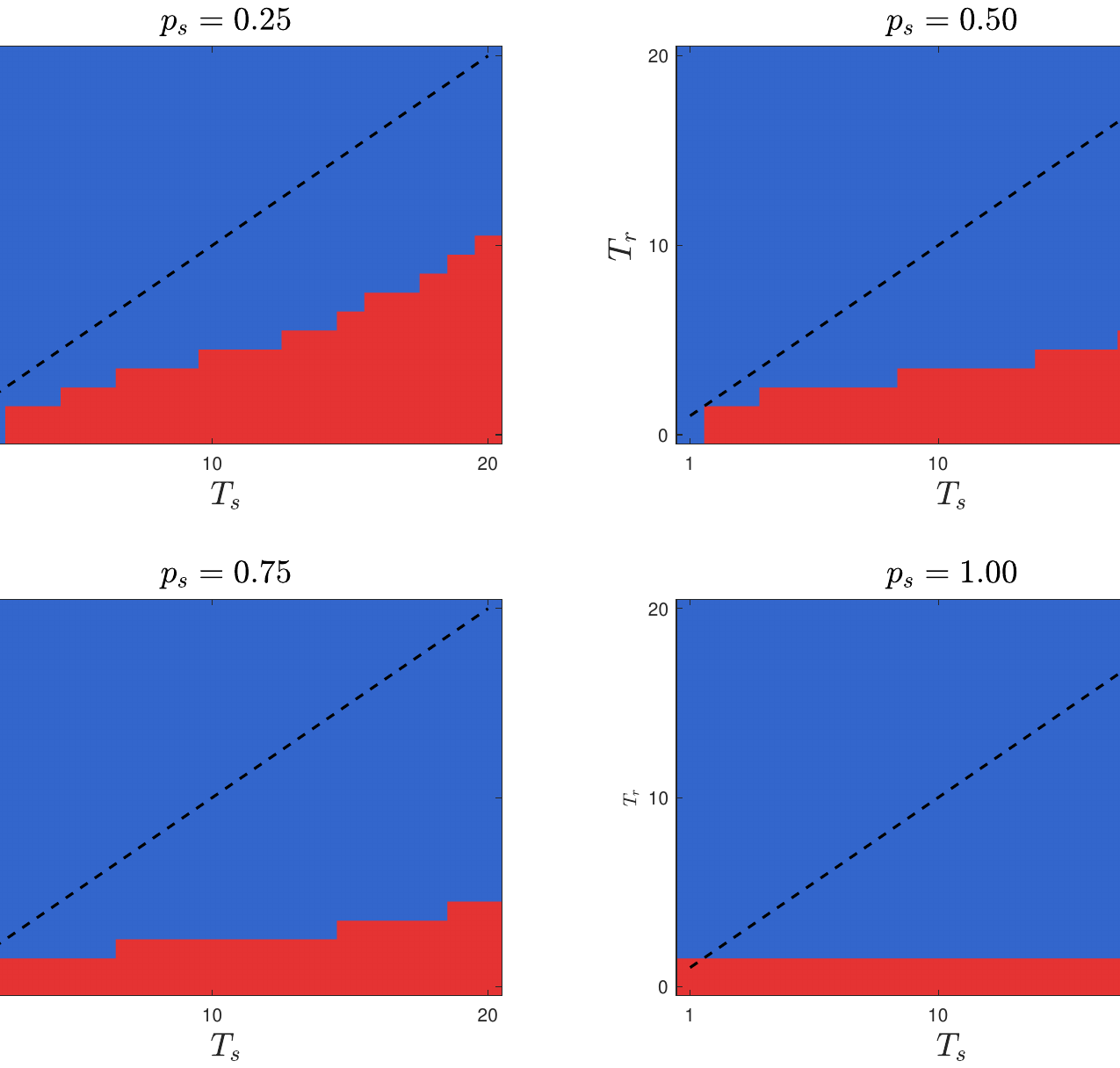}
  \caption{Optimal policy regions for varying $p_s$ values ($K=20$, $c/r=0.5$). 
  \textbf{Blue (shaded) regions:} the sender transmits ($a=1$); 
  \textbf{white regions:} the sender remains silent ($a=0$). 
  The diagonal $T_r = T_s$ marks the boundary above which the receiver 
  already holds fresher data, so transmission is never optimal 
  (Theorem~\ref{thm:global_structure}\,(ii)). As $p_s$ decreases 
  from left to right, the send region expands to compensate for 
  higher packet-loss probability.}
  \label{fig:policy_heatmaps_ps}
\end{figure}

\begin{figure}[t]
  \centering
  \begin{subfigure}[b]{0.49\columnwidth}
    \includegraphics[width=\linewidth, trim=0 210 0 210, clip]{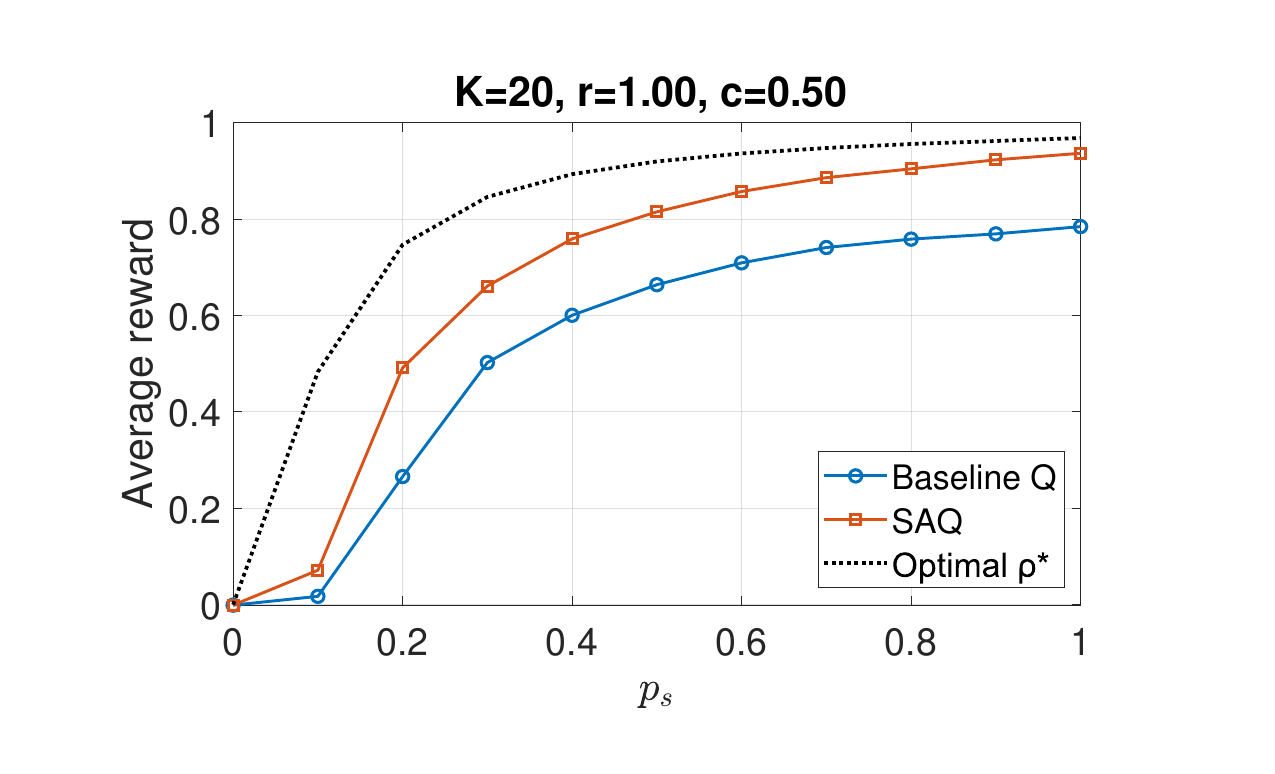}
    \caption{Varying $p_s$ ($c/r=0.5$)}
    \label{fig:sweep_ps}
  \end{subfigure}
  \hfill
  \begin{subfigure}[b]{0.49\columnwidth}
    \includegraphics[width=\linewidth, trim=0 210 0 210, clip]{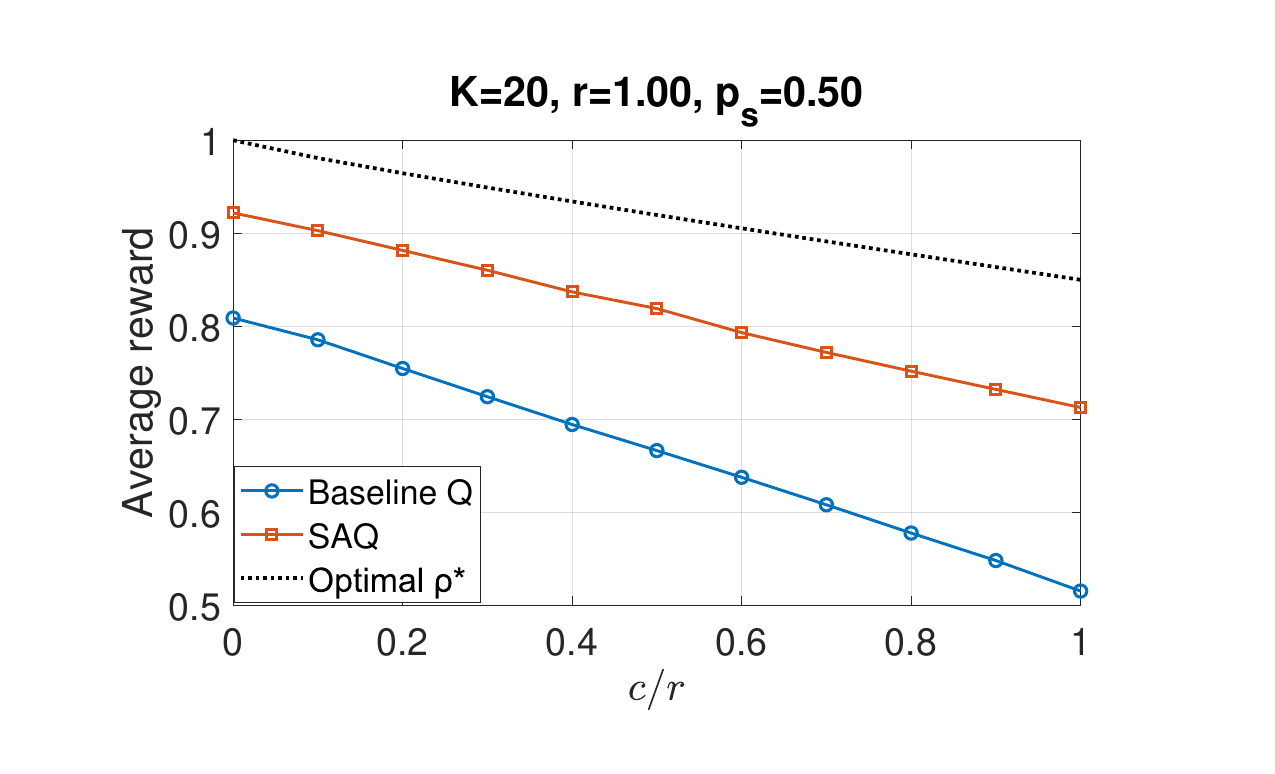}
    \caption{Varying $c/r$ ($p_s=0.5$)}
    \label{fig:sweep_cr}
  \end{subfigure}
  \caption{Average reward vs.\ system parameters for Baseline Q, SAQ, and optimal VI.}
  \label{fig:parameter_sweeps}
\end{figure}
%\begin{figure}[t]
 % \centering
  %\includegraphics[width=0.6\columnwidth]%{Fig6_Policies_vs_cr.eps}
  %\caption{Cumulative average reward: SAQ vs.\ heuristics ($p_s=0.5$, $c/r=0.5$).}
  %\label{fig:policy_comparison}
%\end{figure}
%\begin{figure}[t]
 % \centering
  %\includegraphics[width=0.6\columnwidth]%{Fig6_Policies_vs_cr.eps}
  %\caption{Cumulative average reward: SAQ vs.\ heuristics ($p_s=0.5$, $c/r=0.5$).}
  %\label{fig:policy_comparison}
%\end{figure}
\section{Conclusion}
We studied status updating when samples expire after a finite horizon. A key step is characterizing when a sample loses its value. The scheduling problem reduces to a coupon-collector MDP, and we showed the optimal policy follows a monotone threshold rule, yielding closed-form solutions for constant lifetimes. For random lifetimes with unknown statistics, our structure-aware Q-learning converges faster than standard Q-learning. A natural next step is to identify expiration times for more complex applications and to develop new effectiveness metrics beyond freshness and semantics.
%\newpage
\bibliographystyle{IEEEtran}
\bibliography{bibfile}
\newpage
%================================================================
% APPENDIX
%================================================================
\appendices
% The \section command here creates "A Proof of Theorem..."
% The \label allows the main text to refer to it as "Appendix A"
\section{Proof of Theorem~\ref{thm:optimal-policy}}
\label{app:proof-optimal-policy}

The proof proceeds in five steps: we first establish key structural 
properties of the covariance dynamics, then use a coupling argument 
(exploiting $\beta c \leq 1$) to show that non-admissible policies are 
suboptimal, simplify the objective for admissible policies, show via 
an epoch-ordering argument that JIT achieves the minimum transmission 
rate, and finally compute the JIT value.

\medskip
\noindent\textbf{Step 1: Covariance dominance.}
We first establish that fresher information yields uniformly better 
predictions.

\begin{lemma}[Measurement update reduction]
\label{lem:measurement-update}
For any $R_n \succ 0$, the Kalman measurement update satisfies
$P_{n|n} \preceq P_{n|n-1}$.
\end{lemma}

\begin{proof}
The measurement update is
$P_{n|n} = P_{n|n-1} - P_{n|n-1} H^\top S_n^{-1} H P_{n|n-1}$, where 
$S_n = H P_{n|n-1} H^\top + R_n \succ 0$. This is a congruence 
transformation of $S_n^{-1} \succ 0$, hence positive semidefinite.
\end{proof}

\begin{lemma}[Covariance dominance]
\label{lem:cov-dominance}
For any transmission epochs $m < n$ and any prediction horizon $k \geq 0$,
\begin{equation}\label{eq:cov-dominance}
    P_{n+k|n} \preceq P_{n+k|m}.
\end{equation}
\end{lemma}

\begin{proof}
We proceed by induction on $n - m$.

\emph{Base case} ($n = m + 1$): By Lemma~\ref{lem:measurement-update}, 
$P_{m+1|m+1} \preceq P_{m+1|m}$. For any $k \geq 0$, the Lyapunov 
recursion gives
\begin{align}\label{eq:matrix-bound-compact}
P_{m+1+k|m+1} &= A^k P_{m+1|m+1} (A^k)^\top 
+ \sum_{\ell=0}^{k-1} A^\ell Q (A^\ell)^\top \\
&\preceq A^k P_{m+1|m} (A^k)^\top 
+ \sum_{\ell=0}^{k-1} A^\ell Q (A^\ell)^\top \nonumber \\
&= P_{m+1+k|m}.
\end{align}
where the inequality uses $X \preceq Y \Rightarrow A^k X (A^k)^\top \preceq A^k Y (A^k)^\top$.

\emph{Inductive step}: Suppose the result holds for gap $n - m$. 
For gap $(n+1) - m$, apply the base case with $n$ in place of $m+1$:
$P_{n+1+k|n+1} \preceq P_{n+1+k|n}$.
By the inductive hypothesis with horizon $k + 1$:
$P_{n+1+k|n} = P_{n+(k+1)|n} \preceq P_{n+(k+1)|m} = P_{n+1+k|m}$.
Combining yields $P_{n+1+k|n+1} \preceq P_{n+1+k|m}$.
\end{proof}

\begin{lemma}[Absolute expiration ordering]
\label{lem:absolute-ordering}
For any transmission epochs $m < n$, the absolute expiration times satisfy
\begin{equation}\label{eq:absolute-ordering}
    n + T^{\mathrm{MSE}}_n \geq m + T^{\mathrm{MSE}}_m.
\end{equation}
\end{lemma}

\begin{proof}
By definition of $T^{\mathrm{MSE}}_m$, we have 
$\tr(P_{m+j|m}) \leq \tau_{m+j}$ for $j = 1, \ldots, T^{\mathrm{MSE}}_m - 1$.

For any slot $n + j$ with $n + j \leq m + T^{\mathrm{MSE}}_m - 1$, 
setting $\ell = (n - m) + j$ gives $m + \ell = n + j$ and 
$\ell \leq T^{\mathrm{MSE}}_m - 1$. By Lemma~\ref{lem:cov-dominance}:
\[
    \tr(P_{n+j|n}) \leq \tr(P_{n+j|m}) = \tr(P_{m+\ell|m}) \leq \tau_{m+\ell} = \tau_{n+j}.
\]
Thus the constraint is satisfied at slot $n + j$ under transmission at $n$ 
for all $j = 1, \ldots, (m + T^{\mathrm{MSE}}_m - 1) - n$.

By definition of $T^{\mathrm{MSE}}_n$:
$T^{\mathrm{MSE}}_n \geq (m + T^{\mathrm{MSE}}_m) - n$,
which rearranges to \eqref{eq:absolute-ordering}.
\end{proof}

\medskip
\noindent\textbf{Step 2: Non-admissible policies are suboptimal ($\beta c \leq 1$).}

Given any non-admissible policy $\pi$, we construct an admissible 
policy $\pi'$ satisfying $J(\pi') > J(\pi)$, then invoke Steps~3--4 
to conclude $J(\pi^{\mathrm{JIT}}) \geq J(\pi') > J(\pi)$.

  Define $\pi'$ slot by slot: set 
$a_n^{\pi'} = 1$ if $a_n^{\pi} = 1$, or if $a_n^{\pi} = 0$ and 
$\tr(P^{\mathrm c}_n(\pi')) > \tau_n$.  That is, $\pi'$ mimics $\pi$ 
but additionally transmits at every slot where its own controller 
covariance would otherwise violate~\eqref{eq:MSEcap}.  By 
construction $\tr(P^{\mathrm c}_n(\pi')) \leq \tau_n$ for all~$n$, 
so $\pi'$ is admissible.

  Run $\pi$ and $\pi'$ on the same realization of 
$\{R_n\}$, $\{\tau_n\}$, and $\{\mathbf{w}_n\}$.  Because the 
Kalman filter operates at the sensor independently of the 
transmission decisions, both policies produce the same $P_{n|n}$.
Let $m_n^{\pi}$ and $m_n^{\pi'}$ denote the most recent 
transmission epoch at slot~$n$ under each policy.  Since $\pi'$ 
transmits at least as often as $\pi$, we have 
$m_n^{\pi'} \geq m_n^{\pi}$ for all~$n$.  
Lemma~\ref{lem:cov-dominance} then gives 
$P^{\mathrm c}_n(\pi') = P_{n|m_n^{\pi'}} 
\preceq P_{n|m_n^{\pi}} = P^{\mathrm c}_n(\pi)$, 
so
\begin{equation}\label{eq:coupling-dominance}
\tr(P^{\mathrm c}_n(\pi)) \leq \tau_n 
\;\;\Longrightarrow\;\;
\tr(P^{\mathrm c}_n(\pi')) \leq \tau_n.
\end{equation}

Write  $r_n(\cdot) = \mathds{1}\{\tr(P^{\mathrm c}_n(\cdot)) 
\leq \tau_n\} - \beta c\,a_n(\cdot)$ for the per-slot reward.
We compare $r_n(\pi')$ and $r_n(\pi)$ in three exhaustive cases.
\begin{enumerate}[label=(\roman*),nosep,leftmargin=1.4em]
    \item $\tr(P^{\mathrm c}_n(\pi)) \leq \tau_n$: 
          By \eqref{eq:coupling-dominance}, 
          $\tr(P^{\mathrm c}_n(\pi')) \leq \tau_n$ as well.  
          Since $\pi'$ copies $\pi$'s action in this case 
          ($a_n^{\pi'} = a_n^{\pi}$), we have 
          $r_n(\pi') = r_n(\pi)$.
    \item $\tr(P^{\mathrm c}_n(\pi)) > \tau_n$ and 
          $\tr(P^{\mathrm c}_n(\pi')) \leq \tau_n$: 
          The constraint is violated under~$\pi$ but satisfied 
          under~$\pi'$.  Because $\tr(P^{\mathrm c}_n(\pi)) > \tau_n$ 
          and Assumption~\ref{as:feasibility} guarantees that 
          transmitting restores the constraint, $\pi$ must be 
          silent at slot~$n$ (otherwise the constraint would hold), 
          so $r_n(\pi) = 0$.  
          Policy~$\pi'$ is also silent 
          ($a_n^{\pi} = 0$ and $\tr(P^{\mathrm c}_n(\pi')) \leq \tau_n$), 
          giving $r_n(\pi') = 1$.  Hence $r_n(\pi') - r_n(\pi) = 1$.
    \item $\tr(P^{\mathrm c}_n(\pi)) > \tau_n$ and 
          $\tr(P^{\mathrm c}_n(\pi')) > \tau_n$: 
          By construction $\pi'$ transmits ($a_n^{\pi'} = 1$).  
          After transmission, 
          $\tr(P_{n|n}) \leq \tau_n$ 
          (Assumption~\ref{as:feasibility}), so 
          $r_n(\pi') = 1 - \beta c \geq 0$ since $\beta c \leq 1$.
          Under~$\pi$: if $a_n^{\pi} = 1$, then 
          $r_n(\pi) = 1 - \beta c = r_n(\pi')$; if 
          $a_n^{\pi} = 0$, then $r_n(\pi) = 0 \leq r_n(\pi')$.
\end{enumerate}
In every case $r_n(\pi') \geq r_n(\pi)$.  It remains to show the 
inequality is strict on average.  Because $\pi$ is non-admissible, 
there exists a slot~$n_0$ at which 
$\tr(P^{\mathrm c}_{n_0}(\pi)) > \tau_{n_0}$ and $a_{n_0}^{\pi} = 0$.
At $n_0$, either case~(ii) or case~(iii) applies.  If case~(iii), 
$\pi'$ transmits at~$n_0$; the fresh covariance $P_{n_0|n_0}$ 
satisfies $\tr(P_{n_0|n_0}) \leq \tau_{n_0}$, so at the next 
slot~$n_0{+}1$ we have $\tr(P^{\mathrm c}_{n_0+1}(\pi')) \leq 
\tau_{n_0+1}$ (at least for one step of open-loop prediction 
within the expiration window).  Meanwhile $\pi$, which did not 
transmit at~$n_0$, may still violate the constraint at~$n_0{+}1$.  
This produces a case~(ii) slot at~$n_0{+}1$ with 
$r_{n_0+1}(\pi') - r_{n_0+1}(\pi) = 1 > 0$.  
Since such slots recur with positive frequency, 
$J(\pi') > J(\pi)$.

Because $\pi'$ is admissible, Steps~3--4 give
\[
    J(\pi^{\mathrm{JIT}}) \geq J(\pi') > J(\pi).  \qedhere
\]

\medskip
\noindent\textbf{Step 3: Objective simplification for admissible policies.}

For any admissible policy $\pi$, we have $\Delta_n \geq 1$ for all $n$, 
so $\tr(P^{\mathrm c}_n) \leq \tau_n$ always. Every slot earns reward 1, 
and the objective \eqref{eq:objective} simplifies to:
\begin{equation}\label{eq:admissible-objective}
    J(\pi) = \liminf_{N \to \infty} \frac{1}{N} \mathbb{E}\!\left[
    \sum_{n=0}^{N-1} (1 - \beta c\, a_n) \,\Big|\, P_{0|0}\right]
    = 1 - \beta c \cdot \bar{a}(\pi),
\end{equation}
where $\bar{a}(\pi) \triangleq \limsup_{N \to \infty} 
\frac{1}{N} \mathbb{E}[\sum_{n=0}^{N-1} a_n\mid P_{0|0}]$ 
is the long-run transmission rate. 

Maximizing $J(\pi)$ over admissible policies is thus equivalent to 
minimizing $\bar{a}(\pi)$.

\medskip
\noindent\textbf{Step 4: JIT achieves the minimum transmission rate.}

We show that no admissible policy transmits less frequently than JIT.

Let $\{m_k^{\mathrm{JIT}}\}_{k \geq 1}$ and $\{m_k^\pi\}_{k \geq 1}$ 
denote the transmission epochs under JIT and any admissible policy $\pi$, 
respectively.

\begin{lemma}[Transmission epoch ordering]
\label{lem:epoch-ordering}
For all $k \geq 1$: $m_k^\pi \leq m_k^{\mathrm{JIT}}$.
\end{lemma}

\begin{proof}
We proceed by induction on $k$.

\emph{Base case} ($k = 1$): Both policies must transmit before the 
initial estimate expires. Policy $\pi$ may transmit at or before the 
moment JIT transmits (possibly earlier if $\pi$ transmits ``early'' 
when $\Delta > 1$). Thus $m_1^\pi \leq m_1^{\mathrm{JIT}}$.

\emph{Inductive step}: Suppose $m_{k-1}^\pi \leq m_{k-1}^{\mathrm{JIT}}$.

By Lemma~\ref{lem:absolute-ordering}:
\begin{equation}\label{eq:absolute-ordering-applied}
    m_{k-1}^{\mathrm{JIT}} + T^{\mathrm{MSE}}_{m_{k-1}^{\mathrm{JIT}}} 
    \geq m_{k-1}^\pi + T^{\mathrm{MSE}}_{m_{k-1}^\pi}.
\end{equation}

Under JIT, the next transmission occurs exactly at expiration:
\[
    m_k^{\mathrm{JIT}} = m_{k-1}^{\mathrm{JIT}} + T^{\mathrm{MSE}}_{m_{k-1}^{\mathrm{JIT}}}.
\]

Under $\pi$, the next transmission occurs at or before expiration 
(either at $\Delta = 1$ by necessity, or earlier by choice):
\[
    m_k^\pi \leq m_{k-1}^\pi + T^{\mathrm{MSE}}_{m_{k-1}^\pi}.
\]

Combining with \eqref{eq:absolute-ordering-applied}:
\[
    m_k^\pi \leq m_{k-1}^\pi + T^{\mathrm{MSE}}_{m_{k-1}^\pi} 
    \leq m_{k-1}^{\mathrm{JIT}} + T^{\mathrm{MSE}}_{m_{k-1}^{\mathrm{JIT}}} 
    = m_k^{\mathrm{JIT}}.
\]
\end{proof}

Lemma~\ref{lem:epoch-ordering} implies that for any $N$:
\[
    \bigl|\{k : m_k^\pi \leq N\}\bigr| \geq 
    \bigl|\{k : m_k^{\mathrm{JIT}} \leq N\}\bigr|.
\]
That is, $\pi$ incurs at least as many transmissions as JIT by time $N$. 
Taking limits:
\begin{equation}\label{eq:rate-ordering}
    \bar{a}(\pi) \geq \bar{a}(\pi^{\mathrm{JIT}}),
\end{equation}
with equality if and only if $\pi$ never transmits early (i.e., 
$\pi \equiv \pi^{\mathrm{JIT}}$ on the event that both are admissible).

\medskip
\noindent\textbf{Step 5: JIT value computation.}

Under $\pi^{\mathrm{JIT}}$, transmissions occur at epochs 
$m_1 < m_2 < \cdots$ with inter-transmission times
\[
    L_k \triangleq m_{k+1} - m_k = T^{\mathrm{MSE}}_{m_k}.
\]
Each cycle $k$ has:
\begin{itemize}[nosep]
    \item Length: $L_k = T^{\mathrm{MSE}}_{m_k}$ slots,
    \item Reward: $T^{\mathrm{MSE}}_{m_k}$ (all slots valid) minus 
          $\beta c$ (one transmission),
    \item Net cycle reward: $T^{\mathrm{MSE}}_{m_k} - \beta c$.
\end{itemize}

The transmission process forms a renewal process with cycle lengths 
$\{L_k\}$. By the renewal-reward theorem:
\begin{equation}\label{eq:jit-renewal}
    J(\pi^{\mathrm{JIT}}) = \frac{\mathbb{E}[\text{cycle reward}]}{\mathbb{E}[\text{cycle length}]}
    = \frac{\mathbb{E}[T^{\mathrm{MSE}}_{m_k}] - \beta c}{\mathbb{E}[T^{\mathrm{MSE}}_{m_k}]}
    = 1 - \frac{\beta c}{\bar{T}},
\end{equation}
where $\bar{T} = \lim_{K \to \infty} \frac{1}{K} \sum_{k=1}^{K} T^{\mathrm{MSE}}_{m_k}$ 
equals $\mathbb{E}[T^{\mathrm{MSE}}_{m_k}]$ by the ergodic theorem when 
the limit exists.

The condition $\beta c < \bar{T}$ ensures 
$J(\pi^{\mathrm{JIT}}) > 0$.

\medskip
\noindent\textbf{Conclusion.}

Combining Steps 2--5:
\begin{itemize}[nosep]
    \item Non-admissible policies achieve strictly lower value than 
          $1 - \beta c / \bar{T}$ (Step 2).
    \item Among admissible policies, JIT achieves the minimum transmission 
          rate (Step 4) and hence the maximum value (Step 3).
    \item The JIT value is $1 - \beta c / \bar{T}$ (Step 5).
\end{itemize}
Therefore, $\pi^{\mathrm{JIT}}$ is optimal among all causal policies. \hfill \qed
\section{Proof of Lemma~\ref{lem:monotone}}
\label{app:proof-monotone}

We prove both inequalities by induction on the value iteration sequence.
Define \(\{V^{(k)}\}_{k\ge0}\) by \(V^{(0)}\equiv0\) and
\[
  V^{(k+1)}(T_r,T_s) = \max_{a\in\{0,1\}} Q^{(k+1)}(T_r,T_s,a),
\]
where the Q-function update is
\begin{align*}
Q^{(k+1)}(T_r,T_s,a) &= \bar{R}(s,a) + P_a^{\mathrm{succ}}\,\mathbb{E}_{T_s'}[V^{(k)}(T_s, T_s')] \\
&\quad + P_a^{\mathrm{fail}}\,\mathbb{E}_{T_s'}[V^{(k)}((T_r{-}1)^+, T_s')],
\end{align*}
with $T'_s = \max\{T_s-1, \tilde{T}_s\}$, $\tilde{T}_s \sim p_T$, and
$(P_a^{\mathrm{succ}}, P_a^{\mathrm{fail}}) = (p_s, 1{-}p_s)$ if $a=1$, and $(0,1)$ if $a=0$.

\medskip
\noindent\textbf{Monotonicity in $T_r$.}
Define $\Delta^{(k)}(T_r,T_s,a) := Q^{(k)}(T_r{+}1,T_s,a) - Q^{(k)}(T_r,T_s,a)$.

\emph{Induction hypothesis:} $\Delta^{(k)}(T_r,T_s,a) \ge 0$ for all $(T_r,T_s)$ with $T_r < K$ and all $a$.
The base case $k=0$ holds since $Q^{(0)} \equiv 0$.

\emph{Inductive step:} For action $a=1$:
\begin{align*}
\Delta^{(k+1)}&(T_r,T_s,1) \\
&= \bigl[\bar{R}((T_r{+}1,T_s),1) - \bar{R}((T_r,T_s),1)\bigr] \\
&\quad + (1{-}p_s)\mathbb{E}_{T'_s}\bigl[V^{(k)}(T_r, T'_s) - V^{(k)}((T_r{-}1)^+, T'_s)\bigr].
\end{align*}
The reward difference equals $r(1{-}p_s)\bigl[\mathds{1}\{T_r \ge 1\} - \mathds{1}\{T_r > 1\}\bigr] \ge 0$.
The bracketed expectation is non-negative by the induction hypothesis.
Hence $\Delta^{(k+1)}(T_r,T_s,1) \ge 0$.

For action $a=0$:
\begin{align*}
\Delta^{(k+1)}&(T_r,T_s,0) \\
&= \bigl[\bar{R}((T_r{+}1,T_s),0) - \bar{R}((T_r,T_s),0)\bigr] \\
&\quad + \mathbb{E}_{T'_s}\bigl[V^{(k)}(T_r, T'_s) - V^{(k)}((T_r{-}1)^+, T'_s)\bigr].
\end{align*}
The reward difference is $r\bigl[\mathds{1}\{T_r \ge 1\} - \mathds{1}\{T_r > 1\}\bigr] \ge 0$.
Again, the expectation is non-negative by induction.
Hence $\Delta^{(k+1)}(T_r,T_s,0) \ge 0$.

Since $V^{(k)}(T_r,T_s) = \max_a Q^{(k)}(T_r,T_s,a)$ and the max of non-decreasing 
functions is non-decreasing, we have $V^{(k)}(T_r,T_s) \le V^{(k)}(T_r{+}1,T_s)$.
Taking $k \to \infty$ yields $V(T_r,T_s) \le V(T_r{+}1,T_s)$.

\medskip
\noindent\textbf{Monotonicity in $T_s$.}
An analogous argument, using the fact that larger $T_s$ yields 
$T'_s = \max\{T_s-1, \tilde{T}_s\}$ stochastically larger and 
better next-state values upon success, establishes 
$V(T_r,T_s) \le V(T_r,T_s{+}1)$. \hfill\qed
\section{Proof of Theorem~\ref{thm:row_threshold}}
\label{app:proof-threshold}
The proof relies on monotonicity properties of the action gap 
$\Delta Q(T_r, T_s) = Q(T_r, T_s, 1) - Q(T_r, T_s, 0)$.

\begin{lemma}[Action gap monotonicity]
\label{lem:gap-monotonicity}
$\Delta Q(T_r, T_s)$ is non-increasing in $T_r$ and non-decreasing in $T_s$.
\end{lemma}
\begin{proof}
We first derive an explicit expression for the action gap. Using the 
transition $T_{r,n+1} = T_{s,n}$ upon successful delivery:
\begin{align*}
Q(T_r,T_s,1) &= -c + r\bigl[p_s + (1{-}p_s)\mathds{1}\{T_r > 1\}\bigr] \\
&\quad + \mathbb{E}_{T'_s}\bigl[p_s V(T_s, T'_s) + (1{-}p_s) V((T_r{-}1)^+, T'_s)\bigr], \\
Q(T_r,T_s,0) &= r\,\mathds{1}\{T_r > 1\} + \mathbb{E}_{T'_s}\bigl[V((T_r{-}1)^+, T'_s)\bigr],
\end{align*}
where $T'_s = \max\{T_s - 1, \tilde{T}_s\}$ with $\tilde{T}_s \sim p_T$.
Subtracting yields
\begin{align}
\Delta Q(T_r,T_s) &= -c + r p_s \mathds{1}\{T_r \le 1\} \notag \\
&\quad + p_s \mathbb{E}_{T'_s}\bigl[V(T_s, T'_s) - V((T_r{-}1)^+, T'_s)\bigr].
\label{eq:delta_Q_explicit}
\end{align}

\emph{Monotonicity in $T_r$.}
From \eqref{eq:delta_Q_explicit}, we have
\begin{align*}
&\Delta Q(T_r{+}1,T_s) - \Delta Q(T_r,T_s) \\
&= r\,p_s\bigl[\mathds{1}\{T_r{+}1 \le 1\} - \mathds{1}\{T_r \le 1\}\bigr] \\
&\quad - p_s\,\mathbb{E}_{T'_s}\bigl[V(T_r, T'_s) - V((T_r{-}1)^+, T'_s)\bigr].
\end{align*}
For $T_r = 0$, both the indicator difference and value difference vanish. 
For $T_r = 1$, the indicator difference is $-1$, contributing $-rp_s < 0$, 
while the value term is non-positive by Lemma~\ref{lem:monotone}. 
For $T_r \ge 2$, the indicator difference is zero and the value term 
remains non-positive. Thus $\Delta Q(T_r{+}1,T_s) \le \Delta Q(T_r,T_s)$ 
for all $T_r$.

\emph{Monotonicity in $T_s$.}
Let $T_s' = \max\{T_s - 1, \tilde{T}_s\}$ and $T_s'' = \max\{T_s, \tilde{T}_s\}$ 
for the same realization of $\tilde{T}_s$. Note $T_s'' \ge T_s'$ always. 
From \eqref{eq:delta_Q_explicit}, the difference 
$\Delta Q(T_r,T_s{+}1) - \Delta Q(T_r,T_s)$ reduces to
\[
p_s\,\mathbb{E}_{\tilde{T}_s}\bigl[
  V(T_s{+}1, T_s'') - V((T_r{-}1)^+, T_s'')
  - V(T_s, T_s') + V((T_r{-}1)^+, T_s')
\bigr].
\]
By Lemma~\ref{lem:monotone}, $V$ is non-decreasing in both arguments. 
Since $T_s{+}1 \ge T_s$ and $T_s'' \ge T_s'$, the term 
$V(T_s{+}1, T_s'') - V(T_s, T_s')$ captures the joint benefit 
of a better sender state and a stochastically better next-slot 
arrival.  Meanwhile, the terms $V((T_r{-}1)^+, T_s'') - V((T_r{-}1)^+, T_s')$ 
depend only on the change in $T_s'$, not on $T_s$ itself.  
Thus the net difference satisfies
\begin{align*}
&\bigl[V(T_s{+}1, T_s'') - V(T_s, T_s')\bigr] 
- \bigl[V((T_r{-}1)^+, T_s'') - V((T_r{-}1)^+, T_s')\bigr] \\
&\ge \bigl[V(T_s{+}1, T_s'') - V(T_s, T_s'')\bigr] 
+ \bigl[V(T_s, T_s'') - V(T_s, T_s')\bigr] \\
&\quad - \bigl[V((T_r{-}1)^+, T_s'') - V((T_r{-}1)^+, T_s')\bigr] \\
&\ge 0,
\end{align*}
where the first inequality uses $V(T_s, T_s'') \ge V(T_s, T_s')$ 
(monotonicity in $T_s'$) and the second uses $V(T_s{+}1, T_s'') \ge V(T_s, T_s'')$ 
(monotonicity in the first argument), so the positive terms dominate.
Hence $\Delta Q(T_r,T_s) \le \Delta Q(T_r,T_s{+}1)$.
\end{proof}

We now prove the theorem. Since $\Delta Q(\cdot, T_s)$ is non-increasing 
in $T_r$, its sign changes at most once as $T_r$ increases from $0$ to $K$. 
Define
\[
\theta(T_s) := \max\bigl\{T_r : \Delta Q(T_r,T_s) > 0\bigr\},
\]
and set $\pi^*(T_r, T_s) = 1$ for $T_r \le \theta(T_s)$, and $0$ otherwise. 
The single sign-change property ensures this coincides with the true 
maximizer, establishing the threshold form.

For the monotonicity of thresholds, fix $T_s < K$ and note that 
$\Delta Q(\theta(T_s), T_s) \ge 0$ by definition. Since $\Delta Q$ is 
non-decreasing in $T_s$,
\[
\Delta Q(\theta(T_s), T_s + 1) \ge \Delta Q(\theta(T_s), T_s) \ge 0,
\]
so state $(\theta(T_s), T_s{+}1)$ lies in the ``send'' region, implying 
$\theta(T_s) \le \theta(T_s + 1)$.
\hfill\qed
\section{Proof of Theorem~\ref{thm:global_structure}}
\label{app:proof-global-structure}

\noindent\textbf{(i) Empty receiver, positive net gain.}
Using \eqref{eq:delta_Q_explicit} with $T_r = 0$:
\begin{align*}
\Delta Q(0,T_s) &= -c + r\,p_s\,\mathds{1}\{0 \le 1\} \\
&\quad + p_s\,\mathbb{E}_{T_s'}\bigl[V(T_s, T_s') - V(0, T_s')\bigr] \\
&= (p_s\,r - c)  + p_s\,\mathbb{E}_{T_s'}\bigl[V(T_s, T_s') - V(0, T_s')\bigr].
\end{align*}
Since $p_s\,r - c > 0$ by assumption and $V(T_s, T'_s) \ge V(0, T'_s)$ 
by Lemma~\ref{lem:monotone} (with $T_s \ge 1 > 0$), we have 
$\Delta Q(0,T_s) > 0$. Thus $\pi^*(0,T_s) = 1$. \hfill$\square$

\medskip
\noindent\textbf{(ii) Receiver strictly fresher than sender.}
Fix any state with $T_r > T_s$. Using \eqref{eq:delta_Q_explicit}:
\begin{align}
\Delta Q(T_r,T_s) &= -c + r\,p_s\,\mathds{1}\{T_r \le 1\} \nonumber \\
&\quad + p_s\,\mathbb{E}_{T_s'}\bigl[V(T_s, T_s') - V((T_r{-}1)^+, T_s')\bigr].
\end{align}
Since $T_r > T_s \ge 1$, we have $T_r \ge 2$, so $\mathds{1}\{T_r \le 1\} = 0$ 
and $(T_r - 1)^+ = T_r - 1 \ge 1$.

Also, $T_r > T_s$ implies $T_r - 1 \ge T_s$, so by Lemma~\ref{lem:monotone}:
\[
V(T_s, T'_s) \le V(T_r - 1, T'_s).
\]
Therefore:
\begin{align*}
\Delta Q(T_r,T_s) &= -c + p_s\,\mathbb{E}_{T_s'}\bigl[V(T_s, T_s') - V(T_r{-}1, T_s')\bigr] \\
&= -c + p_s\,\mathbb{E}_{T_s'}\bigl[\underbrace{V(T_s, T_s') - V(T_r{-}1, T_s')}_{\le 0}\bigr] \\
&\le -c < 0.
\end{align*}
Thus $\pi^*(T_r,T_s) = 0$. \hfill$\square$

\medskip
\noindent\textbf{(iii) Monotonicity of transmission.}
Suppose it is optimal to transmit at $(T_r, T_s)$ for some $T_r > 0$, i.e., 
$\Delta Q(T_r, T_s) > 0$. By Lemma~\ref{lem:gap-monotonicity} ($\Delta Q$ non-increasing in $T_r$):
\[
\Delta Q(T_r - 1, T_s) \ge \Delta Q(T_r, T_s) > 0.
\]
Hence it is also optimal to transmit at $(T_r - 1, T_s)$. \hfill\qed
\section{Proof of Theorem~\ref{thm:constant_lifetime}}
\label{app:proof-constant-lifetime}

Fix a threshold policy with parameter $\theta \in \{0,1,\ldots,K\}$. 
A renewal cycle begins each time a transmission succeeds, at which 
point the receiver's timer resets to $T_r = K$.

During slots $T_r = K, K{-}1, \ldots, \theta{+}1$, the sender waits; 
this \emph{waiting phase} lasts $K - \theta$ slots. Once $T_r = \theta$, 
transmission attempts begin and continue until success. Since each 
attempt succeeds independently with probability $p_s$, the 
\emph{attempting phase} has duration $N_A \sim \mathrm{Geom}(p_s)$.

By the renewal-reward theorem, the average reward equals 
$\rho(\theta) = \mathbb{E}[R_{\mathrm{cycle}}] / \mathbb{E}[T_{\mathrm{cycle}}]$. 
The expected cycle length is
\[
\mathbb{E}[T_{\mathrm{cycle}}] 
= (K - \theta) + \frac{1}{p_s}.
\]
For the reward, the waiting phase contributes $r(K-\theta)$ since 
the receiver holds a valid sample throughout. During the attempting 
phase, the sender incurs expected cost $c/p_s$, while the receiver 
earns reward $r$ in slot $i$ only if $T_r = \theta - (i-1) > 0$, 
i.e., for $i \le \theta$. Thus the expected happiness reward is 
$r \cdot \mathbb{E}[\min(N_A, \theta)] = r(1-(1-p_s)^\theta)/p_s$. 
Combining terms,
\[
\rho(\theta) 
= r - \frac{c + r(1-p_s)^\theta}{p_s(K - \theta) + 1}.
\]

To maximize $\rho(\theta)$, define 
$g(\theta) := (K - \theta + 1)(1-p_s)^{\theta-1}$. 
One verifies that $g$ is strictly decreasing in $\theta$, and that 
$\rho(\theta) > \rho(\theta-1)$ if and only if $g(\theta) > c/(p_s r)$. 
Hence the optimal threshold is
\[
\theta^* = \max\bigl\{\theta \ge 1 : g(\theta) > c/(p_s r)\bigr\},
\]
with the convention $\theta^* = 0$ if the set is empty.
\hfill\qed

\section{CartPole Experiment Details}
\label{app:cartpole-details}

This appendix collects the implementation details for the CartPole 
case study of Section~\ref{ssec:cartpole-case}.

\paragraph{Plant dynamics.}
We adopt the classic CartPole environment from the Gymnasium 
library~\cite{towers2023gymnasium}.  The state 
$\mathbf{x}_n = (x_n,\dot{x}_n,\theta_n,\dot{\theta}_n)^\top 
\in \mathbb{R}^4$ comprises cart position, cart velocity, pole angle, 
and angular velocity.  The controller selects a discrete action 
$u_n \in \{0,1\}$ (push left/right) at each step.

\paragraph{Target-tracking objective.}
Rather than the default survival reward, we impose a moving-target 
tracking task, reflecting scenarios where the reference trajectory 
is updated dynamically (e.g., mobile robots receiving waypoints from 
a path planner~\cite{pathak2005velocity}).  The reference position 
$t_n$ follows a mean-reverting random walk
$t_n = (1-\kappa)\,t_{n-1} + w_n$, $w_n \sim \mathcal{N}(0,\sigma_T^2)$,
with $\kappa = 0.03$ and $\sigma_T = 0.02$.  The instantaneous reward is
$r(\mathbf{x},t) = 1 - |x - t|$,
with a crash penalty $r_{\mathrm{crash}} = -50$.  Tracking a 
time-varying reference compounds the effect of staleness: outdated 
observations misrepresent both the current state and its deviation 
from a target that has since moved.

\paragraph{Sensor and channel model.}
A remote sensor produces noisy measurements
\begin{equation}\label{eq:noisy-obs-app}
   \mathbf{y}_n = \mathbf{x}_n + \mathbf{v}_n, \qquad
   \mathbf{v}_n \sim \mathcal{N}(\mathbf{0},\sigma_n^2 I_4),
\end{equation}
where $\sigma_n$ is drawn uniformly from 
$\{0, 0.1, \ldots, 1\}$ (11 levels).  Transmissions cost $c \geq 0$ 
and are assumed reliable ($p_s = 1$) to isolate the scheduling effect.

\paragraph{Controller architecture.}
The controller is a recurrent actor-critic network trained via PPO.  
At deployment it receives potentially stale, noisy observations and 
maintains internal state through a GRU.  
Table~\ref{tab:agent_features_app} lists the input features.

\begin{table}[t]
\centering
\caption{Input features for Controller, Predictor, and Periodic agents.}
\label{tab:agent_features_app}
\begin{tabular}{l l c c c}
\toprule
\textbf{Category} & \textbf{Feature} & \textbf{Ctrl.} & \textbf{Pred.} & \textbf{Per.} \\ 
\midrule
\multirow{2}{*}{\textbf{State}} 
 & Current observation ($\mathbf{y}_n$) & -- & \checkmark & -- \\
 & Stale observation ($\hat{\mathbf{y}}_n$) & \checkmark & -- & -- \\
\midrule
\multirow{3}{*}{\textbf{Context}} 
 & Observation age ($\delta_n$) & \checkmark & -- & \checkmark \\
 & Target position ($t_n$) & \checkmark & \checkmark & -- \\
 & Last action ($u_{n-1}$) & \checkmark & \checkmark & -- \\
\midrule
\textbf{Error} 
 & Tracking error ($x_n - t_n$) & \checkmark & \checkmark & -- \\
\midrule
\textbf{Dynamics} 
 & Observation drift ($y_n - y_{n-1}$) & -- & \checkmark & -- \\
\bottomrule
\end{tabular}
\end{table}

\paragraph{Expiration-time predictor.}
For each tolerance~$\epsilon$, we record expiration times across 
training episodes to build a supervised dataset 
$\mathcal{D}_\epsilon$ pairing input features 
(Table~\ref{tab:agent_features_app}) with the measured 
$T_n$.  A feedforward neural network trained on 
$\mathcal{D}_\epsilon$ predicts $T(\mathbf{y}_n;\epsilon)$ at 
deployment without rollout simulations.  
Code is available at~\cite{Ahmed_Beyond_Freshness_and_2025}.
\end{document}